%% file: main.tex
\documentclass[12pt]{article}

\usepackage[plainpages=false,pdfpagelabels,colorlinks=true,linkcolor=myred,urlcolor=mypurple,citecolor=myblue]{hyperref}
\usepackage{tptstyle}

\usepackage{fullpage}
\usepackage[font={footnotesize,it}]{caption}

\begin{document}

\title{Distributional Conformal Prediction for Markov Processes\thanks{DP's research was supported by NSF grant DMS 24-13718.}}
\author{Dehao Dai$^1$, Kejin Wu $^2$, Dimitris N. Politis$^1$}
  \date{$^1$ University of California, San Diego, 
  $^2$ Loyola University Chicago}
    \maketitle

\begin{abstract}
We introduce the Markov Distributional Conformal Prediction (MDCP) method that extends the distributional conformal prediction (previously developed for regression) to the setting of a strictly stationary Markov process. Instead of relying on a specific model structure to do prediction, the idea of distributional conformal prediction interval aligns with the Model-Free (MF)  Prediction Principle. In analogy to MF prediction of Markov processes, our method exploits the probability integral transform based on estimated transition distribution functions to transform the Markov data to an i.i.d.~dataset. We show a non-asymptotic error bound of MDCPs unconditional coverage rate under a $\beta$-mixing condition and other standard assumptions on the kernel estimators. The asymptotic validity of the conditional prediction interval is also verified. In addition, we show that our conditional prediction interval is still asymptotically valid with Markov processes being $L^p$-$m$-approximable instead of satisfying the mixing property. Numerical simulations and real data experiments are deployed to empirically illustrate the finite-sample performance of MDCP, and compare it with the MF bootstrap prediction method.

\end{abstract} 
\newpage
\input{text}

\FloatBarrier
\bibliographystyle{apalike2}
\bibliography{main.bib}
\newpage

\appendix
\input{supp}

\end{document}

%% file: text.tex
\section{Introduction}

In many applications, such as finance and engineering, it is essential to predict data based on current information. Confidence intervals based on standard $L_2$ point predictors, as well as model-based prediction intervals, have become standard statistical tools within regression frameworks in statistics and machine learning with independent and identically distributed (i.i.d.) data.  However, the i.i.d. assumption is generally invalid for time series regression estimation and prediction, in which the data dependence complicates the theoretical analysis.

Classical time series modeling often relies on linear, nonlinear, or nonparametric frameworks to capture temporal dependence and structural complexity. Common examples include autoregressive (AR), autoregressive moving average (ARMA), and their nonlinear or nonparametric extensions. For example, an AR model can be fitted via ordinary least squares methods in parametric regression problems, or via kernel smoothing or spline regression in nonparametric regression; see \cite{box1976time} for the standard Box-Jenkins method of identifying, fitting, checking and predicting such models. Besides the mean model, ARCH \citep{engle1982autoregressive} and GARCH \citep{bollerslev1986generalized} are two popular models to model the variance within financial time series. Recently, many scholars have studied constructing prediction intervals (PIs) for such models, to quantify uncertainty in forecasts under temporal dependence. For example, \cite{pascual2006bootstrap} proposed a model-based bootstrap approach to build prediction intervals for the GARCH model; \cite{chudy2020long} constructed long-term PIs for time-aggregated future values of univariate economic time series; see also \cite{guo1999multi,pascual2001effects, chen2004nonparametric,manzan2008bootstrap,pan2016bootstrap,politis2023multi,wu2024bootstrap}. However, these methods are only valid under specific model assumptions.
 Having in mind the famous saying by George E. P. Box---``All models are wrong, but some are useful'', it is desired to find an approach to prediction without restrictive model assumptions.

 \textit{Conformal Prediction} (CP) 
 was originally developed by \cite{shafer2008tutorial} as a powerful framework for predictive inference, offering rigorous, finite-sample \textit{unconditional coverage} guarantees for exchangeable data, which is a weaker assumption than i.i.d. The method constructs prediction regions that achieve valid \textit{unconditional coverage} with minimal assumptions, thereby providing a practical and robust alternative to classical statistical inference. Recent developments in statistics and machine learning have expanded this framework to regression, classification, and high-dimensional learning problems, establishing conformal prediction as a general-purpose tool for uncertainty quantification across diverse applications. \cite{romano2019conformalized} proposed a conformalized quantile regression method that provides unconditional coverage guarantees for heteroscedastic observations, while  \cite{sesia2021conformal} developed conformal prediction intervals for non-parametric regression that can automatically adapt to skewed data with both asymptotic marginal and conditional coverage. Some other scholars have made improvements or extensions, such as \cite{lei2013distribution, lei2014distribution, lei2018distribution, xu2021conformal, zaffran2022adaptive}, among others. In particular,  CP can provide the uncertainty quantification for arbitrary prediction algorithms. The better the point predictor, the shorter the prediction interval will be; see \cite{lei2018distribution,tibshirani2023conformal,angelopoulos2024theoretical} for more discussion. 


The previously cited papers on CP for regression are not Model-Free (MF) as their conformity score is usually computed based on the error term once the regression model is estimated. \cite{chernozhukov2021distributional} introduced 
Distributional Conformal Prediction (DCP) applicable to MF regression.
DCP exploits the local probability integral transform (PIT) taken by \citep{politis2013model} to map the vector of regression responses to a vector with i.i.d.~entries; 
then the DCP applies conformal prediction to them. 
An inverse transform brings this method back to the target realm of responses and their prediction.

The MF Prediction Principle  \citep{politis2015model} has also been applied to Markov process data, where the target is one-step-ahead prediction (with prediction intervals) with restrictive time series assumptions. 
In this setting, \cite{pan2015model,  pan2016bootstrap, pan2016bootstrapmarkov} employed the Rosenblatt transformation \cite{rosenblatt1952remarks} based on estimated transition distribution functions to transform the Markov data to an (approximately) i.i.d.~dataset; the latter is then amenable to resampling, leading to the MF bootstrap prediction method.

In the paper at hand, we propose to perform conformal prediction
to the above i.i.d.~dataset that is the outcome of transforming the
Markov data; an inverse transform brings us back to the 
Markov realm and the desired one-step-ahead prediction.
We will call this new approach Markov Distributional Conformal Prediction
(MDCP) because of its analogy to the DCP. 
We will also introduce an extension, namely Predictive MDCP (PMDCP),
that  employs a leave-one-out estimation approach that improves the 
finite-sample coverage rate of prediction intervals.

\subsection{Our contributions}
We summarize our contributions as follows:

\begin{itemize} 
    \item[1.] We propose the so-called Markov Distributional Conformal Prediction (MDCP) method and its variant PMDCP, with the purpose of 
construction one-step-ahead
 Prediction Intervals (PI) for Markov data. Under a $\beta$-mixing condition and other mild assumptions, we provide a non-asymptotic error bound for the unconditional coverage rate of the PIs from MDCP.  
This result helps guide the understanding of the performance of the distributional conformal prediction method with Markov processes. 
    \item[2.] We further show that the conditional coverage rate of MDCP PIs converges to the nominal level asymptotically  under the same 
 $\beta$-mixing
assumption. Moreover, we verify that the asymptotic validity is still sustained assuming that the Markov process is $L^q$-$m$-approximable instead of mixing; see Remark \ref{Remark:differentmeasure} for a discussion on the subtle difference between these two dependence measures. Our proof technique is different from the blocking technique which was applied by \cite{zheng2024conformal,oliveira2024split} to show the asymptotic validity of the conditional prediction interval based on the split conformal prediction method.
    \item[3.] We consider a variant of the MF bootstrap prediction for forecasting Markov processes, which is motivated by the success of applying predictive residuals to alleviate the finite-sample undercoverage issue of the PIs in model-based methods \citep{pan2016bootstrap}. 
Interestingly, the Markov Distributional Conformal Prediction with predictive residuals (PMDCP) does not seem to bring much advantage to the finite-sample coverage rate. However, as revealed by the simulation and real data studies, the lengths of PMDCP intervals conditional on different latest observations of the Markov process are more concentrated than those from MDCP, i.e., they have a smaller standard deviation.
    \item[4.] We derive the uniform error bound of the conditional 
cumulative distribution function (CDF) estimation based on a smoothing kernel for  Markovian data. This result is of independent interest and complements the standard nonparametric inference results for CDF estimation under 
 i.i.d.~data; see Theorem 1.4 in \cite{li2007nonparametric} for the standard result.
    \item[5.] We conduct a comparative study on the performance of the MF bootstrap and conformal techniques with simulated and real datasets. Since the MF prediction principle and the conformal prediction idea were proposed, there have been various extensions in developing new forecasting methods along these two directions, especially for dependent data. 
In addition, the practical suggestions on 
applying these methods are provided in Section \ref{sec: simu}.

\end{itemize}

The paper is organized as follows: In Section \ref{sec: DCP}, we review the framework of DCP and introduce the methodology and algorithms for MDCP and PMDCP. The algorithms of MF and predictive MF prediction methods are also given. The theoretical validation of our methods is presented in Section \ref{sec: asym}. Simulation data and real data studies are conducted in Section \ref{sec: simu}. Technical proofs and additional graphs and tables are in the Appendix.

\subsection{Notations}
In this part, we introduce some notation. For a random variable $X$ and $q\geq 1$, define $\|X\|_q = (\mathbb{E}(|X|^q))^{1/q}$. For a vector $u = (u_1,\ldots, u_d)^\top \in \mathbb{R}^d$, define $\|u\|=\|u\|_w = \sqrt{\sum_{j=1}^d u_j^2}$. Denote $[d]$ as $1,2,3,\ldots, d$ for integer sequence from 1. For two sequences of positive numbers $\{a_n\}$ and $\{b_n\}$, define $a_n \lesssim b_n$ if there exists some constant $C>0$ such that $a_n/b_n \leq C$ as $n \rightarrow \infty$, and define $a_n \asymp b_n$ if $a_n \lesssim b_n$ and $b_n \lesssim a_n$; also define $a_n =O(b_n)$ if there exists a positive constant $C>0$ such that $|a_n|\leq C|b_n|$ for all $n$ and $a_n = o(b_n) $ if $\lim_{n\rightarrow \infty} a_n/b_n =0$. For two random variable sequences $X_n$ and $Y_n$, define $X_n = O_\mathbb{P}(Y_n)$ if for any $0<\varepsilon<1$, there exists a constant $C_\varepsilon$ such that $\sup_{n} \mathbb{P}(|X_n |\geq C_\varepsilon |Y_n|)\leq \varepsilon $; and $X_n = o_\mathbb{P}(Y_n)$ if $X_n/Y_n \stackrel{P}{\rightarrow} 0$. We also use $c, C, C_1, C_2, C_3, \ldots$ to denote positive constants whose values may vary from context to context.

\section{Distributional Conformal Prediction Method}
\label{sec: DCP}
\subsection{Naive conformal prediction intervals}
The basic idea in the conformal prediction is related to sample quantiles. Let $U_1, \ldots, U_n$ be i.i.d random sample. For a given significant level $\alpha \in (0,1)$, another i.i.d sample $U_{n+1}$ 
$$
\mathbb{P}(U_{n+1}\leq \hat{q}_{1-\alpha})\geq 1-\alpha,
$$
where the sample quantile $\hat{q}_{1-\alpha}$ is defined by 
$$
\widehat{q}_{1-\alpha}= \begin{cases}U_{(\lceil(n+1)(1-\alpha)\rceil)} & \text { if }\lceil(n+1)(1-\alpha)\rceil \leq n \\ \infty & \text { otherwise }\end{cases}
$$
and $U_{(1)} \leq \cdots\leq U_{(n)}$ denote the order statistics of $U_1, \ldots, U_n$; $\lceil x \rceil$ represents the smallest integer that is greater than or equal to $x$.

\cite{lei2013distribution}(also \cite{lei2018distribution}) extended this idea to a regression problem $Y = \mu(\bm{X}) + \varepsilon$ given a dataset of $n$ i.i.d or exchangeable observations $\{(\bm{X}_i, Y_i)\}_{i=1}^n$ from a distribution $P$. They proposed the conformal prediction method, which can achieve exact finite-sample marginal (unconditional) validity under exchangeability. The naive $1-\alpha$ prediction interval for a new data point $(\bm{X}_f, Y_f)$ can be defined by
$$
\widehat{C}_{1-\alpha}^{\mathrm{naive}} = \left[\hat{\mu}(\bm{X}_f)-\hat{F}_n^{-1}(1-\alpha),\hat{\mu}(\bm{X}_f)+\hat{F}_n^{-1}(1-\alpha)\right],
$$
where $\hat{\mu}$ is an estimator of the regression function and $\hat{F}_n$ is the empirical distribution of the fitted residuals $|Y_i -\hat{\mu}(\bm{X}_i)|, i = 1, \ldots, n$ and $\hat{F}_n^{-1}(1-\alpha)$ is the $(1-\alpha)$-quantile of $\hat{F}_n$. In general, when the $\hat{F}_n$ is not invertible, the empirical quantile can also be well-defined by
$
\inf\{x\in \mathbb{R}:\hat{F}_n(x)\ge 1-\alpha\}.
$
More discussions can be seen in \cite{lehmann2005testing, lei2013distribution, lei2018distribution}.

\subsection{Model-free prediction principle}
Naive prediction intervals rely on the regression models, which limit the relationship between $\bm{X}$ and $Y$ and require the distribution assumption on $\varepsilon$. More crucially, the coverage of the prediction intervals can not be guaranteed at the nominal level under dependent or unexchangeable data. To alleviate this problem, we will introduce the distributional conformal prediction idea, which shares a similar spirit with the Model-Free (MF) prediction principle. This MF prediction principle has motivated the prediction of financial series by leveraging a specific transformation structure \citep{gulay2018comparison,chen2019optimal,wu2025garchx} and the predictions in various regression settings by involving Deep Neural Networks or assuming a specific data structure \citep{wang2022model,wu2024deep,wang2026model}.

In particular, the MF prediction principle in \cite{politis2013model} is based on probability integral transformation (PIT) to map non-i.i.d datasets to ones consisting of i.i.d variables. This framework can relax the model assumption on the dependence between $\bm{X}$ and $Y$, and it is our main focus throughout this paper. We begin with an assumption on the class of stochastic processes $\{Y_t\}_{t \geq 1}$ considered throughout this paper.

\begin{itemize}
    \item [A1] $\{Y_t\}_{t\geq1}$ forms a strictly stationary and geometrically ergodic Markov process of order $p$ on $(\mathbb{R}, \mathcal{B})$ where $\mathcal{B}$ is a Borel-$\sigma$ algebra over $\mathbb{R}$.
\end{itemize}

\begin{remark}
We should mention that the geometrically ergodic Markov property implies that the series is geometrically $\beta$-mixing with the corresponding dependence coefficient converging to 0 at least exponentially fast as the sample size converges to 0; see Remark \ref{Remark:differentmeasure} for more discussions on different dependence measures. Throughout this paper, we focus on the Markovian data, which can simplify the later transformation process utilized in the MDCP process. From \cite{rosenblatt1952remarks}, the transformation can be more general for non-Markovian. Therefore, the methodologies developed in this paper have the capacity to be extended to a more general scenario. 
\end{remark}

Let $\bm{X}_{t-1} = (Y_{t-1}, \ldots, Y_{t-p})^\top$, given $\bm{X}_{t-1} =\bm{x}$, we can define the conditional distribution of $Y_t$ as
$F(y|\bm{x})= \mathbb{P}(Y_t \leq y|\bm{X}_{t-1}=\bm{x}).$
Now, we consider the data pair $\{(\bm{X}_{t-1}, Y_t)\}_{t=1}^n$ using the PIT. For some positive integer $i\leq p$, define the distributions with partial conditioning as 
$F_i(y |\bm{x}) = \mathbb{P} (Y_t \leq y |\bm{X}_{t-1}^{(i)} = \bm{x})$
where $\bm{X}_{t-1}^{(i)}=(Y_{t-1}, \ldots, Y_{t-i})'$ and $\bm{x} \in \mathbb{R}^i$. Also, we can denote the unconditional distribution as $F_0(y|\bm{X}) = \mathbb{P}(Y_t \leq y)$ which does not depend on $\bm{X}$.

A transformation from our Markov($p$) dataset $Y_1, \ldots, Y_n$ to an $i.i.d$ dataset $V_1, \ldots, V_n$ can be constructed as follows, 
\begin{equation*}
    V_1 = F_0(Y_1|\bm{X}); V_2 = F_1(Y_2|\bm{X}_1^{(1)});\cdots;V_p =F_{p-1}(Y_p|\bm{X}_{p-1}^{(p-1)})
\end{equation*}
\begin{equation*}
    \text{ and } V_t = F(Y_{t}|\bm{X}_{t-1}) \text{ for }t=p+1, p+2, \ldots, n.
\end{equation*}
This is the Rosenblatt transformation that is used in Ch. 8 of \cite{politis2015model} book. Note that the transformation from the vector $(Y_1, \ldots, Y_m)$ to the vector $(V_1, \ldots, V_m)$ is one-to-one and invertible for any $m \leq n$. It implies that $V_1, \ldots, V_n$ are i.i.d Uniform$(0,1)$. This is also an application of \cite{rosenblatt1952remarks} transformation in the case of Markov($p$).
To apply the transformation mentioned before, we also make the following assumption:
\begin{itemize}
    \item [A2] The marginal density $f_x$ of $\bm{X}$ is positive and has continuous second-order derivatives, and the conditional CDF $F(y|\bm{x})$ has continuous second-order derivatives with respect to $\bm{x}$. $F$ is Lipschitz continuous in $y$ for all $\bm{x}$; the conditional density $f(y|\bm{x})$ is positive for all $y$ and $\bm{x}$. 
\end{itemize}
Based on the PIT transformation, we define the kernel-based smoothed conditional CDF $\hat{F}(Y_t|\bm{X}_{t-1})$ with the pair data $\{(\bm{X}_{t-1}, Y_t)\}_{t=p+1}^n$ by
\begin{equation}
\label{Eq:kernelEst}
 \hat{F}(Y_t|\bm{X}_{t-1}) =\frac{\frac{1}{n-p}\sum_{i=p+1}^{n} W_h\left(\bm{X}_{i-1}, \bm{X}_{t-1}\right) K\left(\frac{Y_t-Y_i}{h_0}\right)}{\overline{W}_h(\bm{X}_{t-1})};
\end{equation}
and the transformed data $\{V_t: p+1 \leq t\leq n\}$ can be calculated by $
    V_t = \hat{F}(Y_t|\bm{X}_{t-1}); $
where $W_h\left(\bm{X}_{i-1}, \bm{X}_{t-1}\right)=\prod_{s=1}^p \frac{1}{h_s}w\left(\frac{X_{i-1,s}-X_{t-1,s}}{h_s}\right)$ and $\overline{W}_h(\bm{X}_{t-1})=\frac{1}{n-p}\sum_{i= p+1}^{n} $ $W_h\left(\bm{X}_{i-1}, \right.$ $\left.\bm{X}_{t-1}\right)$; $w(\cdot)$ is a univariate, symmetric kernel density function with bounded support; $X_{i-1,s}$ and $X_{t-1,s}$ are the $s$-th coordinate of $\bm{X}_{i-1}$ and $\bm{X}_{t-1}$; $K$ is a smooth distribution function that is strictly increasing, a proper CDF; $h$ and $h_0$ are bandwidths. For simplicity, we take $h=h_1=h_2=\cdots=h_p$.

Besides, bootstrap, originally introduced by \cite{efron1992bootstrap} for i.i.d data, as a nonparametric method, can estimate the quantile of roots between the target and the estimated sample functional by generating a large number of samples. Under this spirit, we summarize the MF bootstrap prediction procedure from \cite{pan2016bootstrapmarkov} in Algorithm \ref{algo:mfb}.
\begin{algorithm}[!ht]
\caption{Model-Free (MF) Bootstrap Method for Markov$(p)$}
\begin{algorithmic}[1]
\Require Data $ \{(\bm{X}_{t-1},Y_t)\}_{t=1}^n$. Some large positive integer $M$. Bootstrap size $B$. Significant level $\alpha$.
\State Use \eqref{Eq:kernelEst} to obtain the $V_t = \hat{F}(Y_{t}|\bm{X}_{t-1})$ for $t = p+1, \ldots, n$.
\State Calculate $\hat{Y}_{n+1}$, the predictor of $Y_{n+1}$ by the sample mean
$$
\hat{Y}_{n+1} = \frac{1}{n-p}\sum_{t=p+1}^n \hat{F}^{-1}(V_t|\bm{X}_n).
$$
\For {step $i \in \{1,\ldots, B\}$}
\State Resample randomly with replacement the transformed data $V_{p+1}, \ldots, V_n$ to create the data $V_{-M}^*, V_{-M+1}^*, \ldots, V_0^*, V_1^*, \ldots, V_n^*, V_{n+1}^*$.
\State Draw $Y^*_{-M}, \ldots, Y_{-M+p-1}^*$ from any consecutive $p$ values of the dataset $(Y_1, \ldots, Y_n)$; Let $\bm{X}^*_{-M+p-1} = (Y^*_{-M+p-1}, \ldots, Y^*_{-M})$. Especially, for $p=1$, $X^*_{-M} = Y^*_{-M}$ is drawn from $(Y_1, \ldots, Y_n)$.
\State Generate $Y_t^* = \hat{F}^{-1}(V_t^*|\bm{X}_{t-1}^*)$ for $t = -M +p, \ldots, n$.
\State Calculate the bootstrap future value $Y^*_{n+1} = \hat{F}^{-1}(V^*_{n+1}|\bm{X}_n)$.
\State Calculate the bootstrap mean $\hat{Y}^*_{n+1} = \frac{1}{n-p}\sum_{t=p+1}^n \hat{F}^{*-1}(V_t^*|\bm{X}_n)$ where 
$$
 \hat{F}^*(Y|\bm{X}) =\frac{\frac{1}{n-p}\sum_{t=p+1}^{n} W_h\left(\bm{X}_{t-1}^*, \bm{X}\right) K\left(\frac{Y-Y_t^*}{h_0}\right)}{\overline{W}_h(\bm{X})}.
$$
\State Calculate the bootstrap root $Y_{n+1}^* -\hat{Y}^*_{n+1} $.
\EndFor
\State The $B$ bootstrap root replicates are collected in the form of an empirical distribution whose $\alpha$-quantile is denoted $q(\alpha)$.
\Statex \textbf{Output:} The $(1-\alpha)\times 100\%$ equal-tailed predictive interval for $Y_{n+1}$ is given by
$$
\hat{C}_{1-\alpha}^{\mathrm{MF}}(\bm{X}_{n})=\left[\hat{Y}_{n+1}+q(\alpha/2), \hat{Y}_{n+1}+q(1-\alpha/2)\right].
$$
\end{algorithmic}
\label{algo:mfb}
\end{algorithm}

\begin{remark}
    In step 8 of Algorithm \ref{algo:mfb}, we could generate a larger number of $V^*_t$ in alignment with the idea of bootstrap estimation. However, we obey the procedure taken by \cite{pan2015model, pan2016bootstrapmarkov}, in which the confidence interval is built for the conditional mean of the Markov process.
\end{remark}
\begin{remark}
    Original from $Y_t$ for $t = 1, \ldots, n$, we may have the transformed data $V_1, \ldots, V_n$, but here we just consider full paired data $(Y_t, \bm{X}_{t-1})$ rather than truncated $(Y_t, \bm{X}_{t-1}^{(i)})$ for $i = 0, \ldots, p-1$. Although we lose $p$ number of transformed variables, $(V_{p+1}, \ldots, V_n)$ will be more appropriate to do further predictions rather than $V_1, \ldots, V_n$ since the first $p$ transformation may not be accurate.
\end{remark}

\begin{remark}
    $M$ in the Algorithm \ref{algo:mfb} is the number of ``padded history'' in Step 4. In general, this technique is commonly used to generate a stationary time series, so-called ``warm-up'' or ``burn-in'' sequences. It also guarantees the sequence cannot be affected by the initial state. 
\end{remark}
Since $Y_{n+1}$ is not observed, the above estimated conditional distribution for $Y_{n+1}$ treats the pair $\left(X_n, Y_{n+1}\right)$ as an ``out-of-sample'' pair. To mimic this situation in the MF set-up, we can use the trick of \cite{pan2016bootstrap}, i.e., to calculate an estimate of $F(Y_t|\bm{X}_{t-1})$ based on a dataset that excludes the pair $\left(\bm{X}_{t-1}, Y_t\right)$ for $t=p+1, \cdots, n$. In other words, define the ``leave-one-out'' estimator
\begin{equation}
\label{Eq:KernelEstPred}
\hat{F}_t(Y_t|\bm{X}_{t-1}) =\frac{\frac{1}{n-p}\sum_{i=p+1, i \neq t }^{n} W_h\left(\bm{X}_{i-1}, \bm{X}_{t-1}\right) K\left(\frac{Y_t-Y_i}{h_0}\right)}{\tilde{W}_h(\bm{X}_{t-1})} ,\text{ for } t = p+1, \ldots, n;
\end{equation}
where $\tilde{W}_h(\bm{X}_{t-1}) = \sum_{k = p+1, k\neq t}W_h(\bm{X}_{k-1}, \bm{X}_{t-1})$
to construct the transformed data:
$$
\widetilde{V}_t=\hat{F}_t(Y_t|\bm{X}_{t-1}) \text{ for } t= p+1, \ldots, n.
$$
Here, the $\widetilde{V}_t$ serve as the analogs of the predictive residuals studied in \cite{pan2016bootstrap} in a nonparametric autoregression setup. Algorithm \ref{algo:pmfb} tabulated the procedure to construct the Predictive Model-free Bootstrap prediction interval; see Section 6.2 of \cite{pan2016bootstrapmarkov}.
\begin{algorithm}[!ht]
\caption{Predictive Model-Free (PMF) Bootstrap Method for Markov$(p)$}
\begin{algorithmic}[1]
\Require  Data $ \{(\bm{X}_{t-1},Y_t)\}_{t=1}^n$. Some large positive integer $M$. Bootstrap size $B$. Significant level $\alpha$.
\State Use \eqref{Eq:KernelEstPred} to obtain the $\tilde{V}_t = \hat{F}_t(Y_{t}|\bm{X}_{t-1})$ for $t = p+1, \ldots, n$.
\State Calculate $\hat{Y}_{n+1}$ by step 2 in Algorithm \ref{algo:mfb} with $\widetilde{V}_t$ series.
\For {step $i \in \{1,\ldots, B\}$}
\State Resample randomly with replacement the transformed data $\tilde{V}_{p+1}, \ldots, \tilde{V}_n$ to create the data $\tilde{V}_{-M}^*, \tilde{V}_{-M+1}^*, \ldots, \tilde{V}_0^*, \tilde{V}_1^*, \ldots, \tilde{V}_n^*, \tilde{V}_{n+1}^*$.
\State Draw $\tilde{Y}^*_{-M}, \ldots, \tilde{Y}_{-M+p-1}^*$ from any consecutive $p$ values of the dataset $(Y_1, \ldots, Y_n)$; Let $\tilde{\bm{X}}^*_{-M+p-1} = (\tilde{Y}^*_{-M+p-1}, \ldots, \tilde{Y}^*_{-M})$. Especially, for $p=1$, $\tilde{\bm{X}}^*_{-M} = \tilde{Y}^*_{-M}$ is drawn from $(Y_1, \ldots, Y_n)$.
\State Calculate the bootstrap future value $\tilde{Y}^*_{n+1} = \hat{F}^{-1}(\tilde{V}^*_{n+1}|\bm{X}_n)$.
\State Calculate the bootstrap mean $\hat{Y}^*_{n+1}$ as similar to steps 1 to 6 and step 8 in Algorithm \ref{algo:mfb} with $\tilde{Y}^*_{t}$ series.
\State Calculate the bootstrap root $\tilde{Y}^*_{n+1} -\hat{Y}^*_{n+1} $.
\EndFor
\State The $B$ bootstrap root replicates are collected in the form of an empirical distribution whose $\alpha$-quantile is denoted $q(\alpha)$.
\Statex\textbf{Output:} The $(1-\alpha)\times 100\%$ equal-tailed predictive interval for $Y_{n+1}$ is given by
$$
\hat{C}_{1-\alpha}^{\mathrm{PMF}}(\bm{X}_{n})=\left[\hat{Y}_{n+1}+q(\alpha/2), \hat{Y}_{n+1}+q(1-\alpha/2)\right].
$$
\end{algorithmic}
\label{algo:pmfb}
\end{algorithm}

In Section \ref{sec: asym}, we will give the technical Lemma \ref{lem: mse of uhat} to verify the asymptotic validity of the MF/PMF prediction interval. In short, the crucial condition is the uniform consistency property of the conditional CDF estimator, where the mixing condition is unavoidable in the proof. However, there are no non-asymptotic results for MF bootstrap based methods. In contrast, conformal prediction enjoys exact finite-sample unconditional coverage under the i.i.d. setting, but fails once dependence is present. To bridge this gap, we integrate these two approaches: using the MF idea to adapt CP to data with Markov dependence.


\subsection{Distributional conformal prediction}

We first introduce the classical CP framework in the \cite{shafer2008tutorial}. The framework provides a general methodology for constructing PIs with unconditional coverage guarantees under the assumption of data exchangeability. For each candidate $Y_f \in \mathbb{R}$, we construct an augmented regression estimator $\hat{\mu}$ using the augmented data $\{(\bm{X}_t, Y_t)\}_{t=1}^n \cup (\bm{X}_f, Y_f)$. Now define the fitted residuals
$$
R_{Y_f, t} = |Y_t -\hat{\mu}(\bm{X}_t)|,\quad t = 1, \ldots, n; \qquad R_{Y_f, n+1} =|Y_f - \hat{\mu}(\bm{X}_f)|,
$$
and the rank of $R_{y, n+1}$ by
$$
\pi(Y_f) = \frac{1}{n+1}\sum_{t=1}^{n+1}\mathbbm{1}\{R_{Y_f, t} \leq R_{Y_f, n+1}\} = \frac{1}{n+1}+\frac{1}{n+1}\sum_{t=1}^{n}\mathbbm{1}\{R_{Y_f, t} \leq R_{Y_f, n+1}\}.
$$
By exchangeability of the data points and symmetry of $\hat{\mu}$, the ranks $\pi(Y_f) $ evaluated at $Y_f$ provides $\pi(Y_f)$ is discrete uniform over $\{1/(n+1), 2/(n+1), \ldots, 1\}$. The $1-\alpha$ prediction interval can be computed as 
$$
\hat{C}_{1-\alpha}^{\mathrm{CP}}(\bm{X}_f) = \{Y_f \in \mathbb{R}:\pi(Y_f) \leq \frac{1}{n+1}\lceil(1-\alpha)(n+1)\rceil\}.
$$

Now we use the MF principle to adapt CP to Markov processes, resulting in a so-called Markov Distributional Conformal Prediction (MDCP), which can overcome the limitation of standard CP.

For any candidate values $y \in \mathbb{R}$, our ultimate goal is to test whether it belongs to the prediction inference of $Y_f$ conditional on a new data point $\bm{X}_f$. In the dependent process case, we can take $\bm{X}_f = \bm{X}_{n}$ and $Y_f = Y_{n+1}$. Due to the stationarity, we denote the distribution of $Y_t$ conditional on $\bm{X}_{t-1}$ by $F(Y_t|\bm{X}_{t-1})$ for any $t$. Under the distributional CP idea, based on the augmented data $\{(\bm{X}_{t-1}, Y_t)\}_{t=p}^n \cup (\bm{X}_{n}, Y_{n+1})$, we estimate the ground-truth conditional distribution function $F(\cdot |\bm{X})$ by $\hat{F}^{(Y_{n+1})}(\cdot| \bm{X})$ with candidate $Y_{n+1}$. In practice, we can consider a $\mathcal{Y}_\text{trial}$, s.t., all plausible values $y\in \mathcal{Y}_\text{trial}$ are collected to make the prediction interval. Given $Y_{n+1} := y$, based on the CDF estimator \eqref{Eq:kernelEst} to get the corresponding ranks $\{\hat{U}^{(Y_{n+1})}_t\}_{t=p+1}^{n+1}$ which are just conditional CDF estimators with augmented data, i.e.,
\begin{align}
\label{Eq: ranks}
    \hat{U}_{t}^{(Y_{n+1})} = \begin{cases}
        \hat{F}^{(Y_{n+1})}(Y_t|\bm{X}_{t-1}) & \text{  if  } p+1 \leq t \leq n\\
         \hat{F}^{(Y_{n+1})}(Y_{n+1}|\bm{X}_{t-1}) & \text{  if  } t = n+1\\
    \end{cases};
\end{align}
where
$$
\hat{F}^{(Y_{n+1})}(Y_t|\bm{X}_{t-1}) = \frac{\frac{1}{n-p + 1}\left(\sum_{i=p+1}^{n} W_h\left(\bm{X}_{i-1}, \bm{X}_{t-1}\right) K\left(\frac{Y_t-Y_i}{h_0}\right) + W_h(\bm{X}_{n},\bm{X}_{t-1})K\left(\frac{Y_t-Y_{n+1}}{h_0}\right)\right)}{ \frac{1}{n-p + 1} \left(\sum_{i=p+1}^{n} W_h\left(\bm{X}_{i-1}, \bm{X}_{t-1}\right) + W_h(\bm{X}_{n},\bm{X}_{t-1})\right)}.
$$
Then, we can define the conformity scores as $\hat{V}_t^{(Y_{n+1})} = |\hat{U}_t^{(Y_{n+1})} - 1/2|$, and the $p$-values are obtained by 
\begin{equation}\label{Eq:p-v}
\hat{p}(Y_{n+1}) = \frac{1}{n-p+1}\sum_{t=p+1}^{n+1}\mathbbm{1}\{\hat{V}_t^{(Y_{n+1})} \geq \hat{V}_{n+1}^{(Y_{n+1})}\}.
\end{equation}
Then, the p-value can be used to guide the building of PI. We summarize this approach in Algorithm \ref{algo: FDCPMarkov}.
\begin{algorithm}
\caption{MDCP: Full DCP algorithm for Markov$(p)$}\label{algo: FDCPMarkov}
\begin{algorithmic}[1]
\Require (a) Data $\{\bm{X}_{t-1},Y_t\}_{t=p+1}^n$; (b) Significant level $\alpha \in (0,1)$ (c) Test values $\mathcal{Y}_\text{trial}$

\For{$Y_{n+1}\in \mathcal{Y}_\text{trial}$}
\State Define the augmented data $\{\bm{X}_{t-1}, Y_t\}_{t=p+1}^n \cup (\bm{X}_{n}, Y_{n+1})$.
\State Choose optimal bandwidth $h$ and $h_0$ to calculate ranks defined as \eqref{Eq: ranks}.
\State Compute $\hat{p}(Y_{n+1})$ as in \eqref{Eq:p-v}.
\EndFor
\Statex \textbf{Output}: Return $(1-\alpha)$ PI $\widehat{C}_{1-\alpha}^{\mathrm{MDCP}}(\bm{X}_n) = \{Y_{n+1} \in \mathcal{Y}_\text{trial}: \hat{p}(Y_{n+1}) >\alpha\}.$
\Statex
\textbf{Comments:} We can obtain the closed form of PI: $[\min (\widehat{C}_{1-\alpha})(\bm{X}_n), \max(\widehat{C}_{1-\alpha}(\bm{X}_n))]$.
\end{algorithmic}
\end{algorithm}

\begin{remark}[On Bandwidth Choice]
  When $p = 1$, the optimal smoothing of $\hat{F}^{(Y_{n+1})}(u|x)$ with respect to Mean Squared Error (MSE) in \cite{li2007nonparametric} requires that $h_0 = O_\mathbb{P}(n^{-2/5})$ and $h = O_\mathbb{P}(n^{-1/5})$. For general $p$, $h = O_\mathbb{P}(n^{-1/{(4+p)}})$ and $h_0 = O_\mathbb{P}(n^{-2/(4+p)})$. In the practice, we chose $h$ and $h_0$ through cross-validation simultaneously with {\tt R} function {\tt npcdistbw} from {\tt np} package. This function can also estimate $h$ and $h_0$ under multivariate CDFs. To simplify the technical proof, we further take $h_0 \asymp h$.
\end{remark}

\begin{remark}[$\mathcal{Y}_\text{trial}$ Choice]
\label{rem: choice}
We can choose $\mathcal{Y}_\text{trial}$ to be a fine grid between $- \max_{1\le t\le n} |Y_t |$ and $\max_{1\le t \le n} |Y_t |$. Besides, we can get a fine grid between the training dataset of $\bm{X}$. If $|Y_t|$ is bounded and satisfies Assumption A1, we can have $\mathbb{P}(|Y_{n+1}|\geq \max_{1\leq t\leq n}|Y_t|) \leq n^{-\alpha}$ for some $\alpha \geq 1$. Specifically, if the data is exchangeable, we can have the exact probability upper bound $\mathbb{P}(|Y_{n+1}|\geq \max_{1\leq t\leq n}|Y_t|) \leq n^{-1}$ \citep{chen2016trimmed}.
\end{remark}

We also adapt the predictive principle to ``leave-one-out'' estimator and construct the predictive conditional ranks based on \eqref{Eq:KernelEstPred}
\begin{align}
\label{Eq: pred ranks}
    \tilde{U}_{t}^{(Y_{n+1})} = \begin{cases}
        \hat{F}^{(Y_{n+1})}_t (Y_t|\bm{X}_{t-1}) & \text{  if  } p+1 \leq t \leq n\\
         \hat{F}^{(Y_{n+1})}_t (Y_{n+1}|\bm{X}_{n}) & \text{  if  } t = n+1\\
    \end{cases},
\end{align}
where 
$
\hat{F}_t^{(Y_{n+1})}(Y_t|\bm{X}_{t-1}) = \frac{\frac{1}{n-p + 1} \sum_{i=p+1, i\neq t}^{n+1} W_h\left(\bm{X}_{i-1}, \bm{X}_{t-1}\right) K\left(\frac{Y_t-Y_i}{h_0}\right) }{ \frac{1}{n-p + 1} \sum_{i=p+1, i\neq t}^{n+1} W_h\left(\bm{X}_{i-1}, \bm{X}_{t-1}\right) }.
$
Similarly, we can obtain the $p$-values $\hat{p}(Y_{n+1}) = \frac{1}{n-p+1}\sum_{t=p+1}^{n+1} \mathbbm{1}\{\tilde{V}_t^{(Y_{n+1})}\ge \tilde{V}_{n+1}^{(Y_{n+1})}\}$ with the conformity scores $\tilde{V}_t^{(Y_{n+1})} = |\tilde{U}_t^{(Y_{n+1})} -1/2|$. This leads to the Predictive DCP algorithm for the Markov process, the so-called PMDCP method.

\begin{algorithm}[!ht]
\caption{Predictive MDCP (PMDCP) algorithm for Markov$(p)$}\label{algo: PredDCPMarkov}
\begin{algorithmic}[1]
\Statex
The algorithm is identical to Algorithm \ref{algo: FDCPMarkov} except that ranks change to \eqref{Eq: pred ranks}.
\Statex
\textbf{Comments:} We can obtain the closed form of PI: $\hat{C}^{\mathrm{PMDCP}}_{1-\alpha}=[\min (\widehat{C}_{1-\alpha})(\bm{X}_n), \max(\widehat{C}_{1-\alpha}(\bm{X}_n))]$.
\end{algorithmic}
\end{algorithm}

\section{Asymptotic Validity of PIs}
\label{sec: asym}
For i.i.d case or exchangeable data $\{(\bm{X}_t,Y_t)\}_{t=1}^n \cup(\bm{X}_f, Y_f)$, the coverage guarantee for the naive conformal prediction \citep{tibshirani2023conformal} provides the unconditional validity, i.e.,
$$
1-\alpha \leq \mathbb{P}(Y_f \in \widehat{C}_{1-\alpha}^{\mathrm{naive}}(\bm{X}_f)) \leq 1-\alpha +\frac{1}{n+1}, 
$$
and asymptotic conditional validity in the distributional conformal prediction style
$$
\mathbb{P}(Y_f \in \widehat{C}_{1-\alpha}^{\mathrm{naive}}(\bm{X}_f) |\bm{X}_f) \rightarrow 1-\alpha  \text{ as } n \rightarrow \infty;
$$

We next investigate whether the two guarantees still hold in the case of predicting a Markov process with MDCP and PMCDP. First, we introduce some regular assumptions:
 \begin{itemize}
    \item [A3] Let $N = n-p+1$, $h_0 \asymp h \to 0$, $Nh^{p}\to \infty$ and $ph^2 \to 0$ when $n \to \infty$.
     \item [A4] The domains of $Y$ and $\bm{X}$ are bounded by $C_M>0$, i.e., $(\bm{X},Y)\in \mathcal{X}:=[-C_M,C_M]^{p+1}$ with positive $C_M$, an arbitrarily large constant.
    \item [A5] $w(\cdot)$ and $K(\cdot)$ are positive, differentiable, and bounded on their whole domains with 
    \begin{itemize}
        \item [(i)] $\int w(v)dv =1$, $\int vw(v)dv = 0$ and $\int v^2 w(v)dv =C_w<\infty.$
        \item [(ii)] $\int K'(z)dz = 1$, $\int zK'(z)dz = 0$ and $\int z^2K'(z)dz = C_K<\infty.$
    \end{itemize}
     Furthermore, we assume $\prod_{s=1}^p w(\cdot)$ and $K(\cdot)$ satisfy Lipschitz continuity with constant $L_{w}$ and $L_{K}$, respectively. 
\end{itemize}

\begin{remark}
Assumption A4 states that the domains of $Y$ and each dimension of $\bm{X}$ are bounded by an arbitrarily large constant $C_M$. This is for technical convenience to simplify the proof. In fact, the kernel CDF $K$ and density $w$ can be allowed to the full unbounded space $\mathbb{R}$. In practice, $C_M$ is chosen as $\max_{1\leq i\leq n} |Y_i|$ which has been discussed in $\mathcal{Y}_{\mathrm{trial}}$ choice from Remark \ref{rem: choice}.
\end{remark}
The crucial step in the DCP-type method is the estimation of the CDF. Traditionally, CDF is estimated by the empirical distribution function, which is not smooth, though it has good statistical properties. The usual method for nonparametric PDF or CDF estimation is the kernel method. There is a considerable literature to study kernel-type CDF estimation. \cite{azzalini1981note} studied univariate CDF estimation, and \cite{jin1999kernel} extended the one-dimensional estimation to the multivariate case. They provided the order of convergence rate of one smoothing parameter in estimating CDF as $O(n^{-1/3}).$ \cite{liu2008kernel} proved the smooth multivariate kernel estimator converges to the true CDF at a rate of $O(n^{-1/2}\log n)$ by using a diagonal bandwidth matrix. Recently, \cite{mansouri2024nonparametric} proposed the bivariate kernel CDF estimator with Birnbaum–Saunders distribution in \cite{mombeni2021asymmetric} to obtain consistent estimators free from boundary bias.

We then focus on the uniform asymptotic bound for kernel-based CDF on the augmented dependent data. We first give the following key Lemma \ref{lem: mse of uhat}, which is the basis for showing the asymptotic validity of MF/PMF and MDCP/PMDCP prediction intervals. The proof is inspired by Theorem 6.2 of \cite{li2007nonparametric}. However, it has a fundamental difference since the data are dependent in our case, so that some concentration inequality for multivariate dependent data pairs, e.g., Bernstein’s inequality, is necessary to show the uniform consistency of our smoothing kernel estimator.

\begin{lemma}
\label{lem: mse of uhat}
    Under A1 to A5, we have
    $$
    \sup_{\bm{x},y}|\widehat{F}(y|\bm{x}) - F(y|\bm{x})| =O\left(\frac{\log N}{\sqrt{Nh^{p}}}+h^2\right),
    $$ 
    with probability at least $1 - S_n$, where $S_n$ is one appropriate sequence converges to 0; see Remark \ref{Remark:Sn} for the discussion about this sequence.
\end{lemma}

\begin{remark}[Discussion of sequence $S_n$]\label{Remark:Sn}
Denote $h^p = N^{-\tau}$ for $0<\tau<1$. From the proof of Lemma \ref{lem: mse of uhat}, $S_n$ is taken as $1 - 2J_nN^{-C_9} - 2J_nN^{-C_9^{'}}$, where $C_9>0$ and $\tau<$ are appropriate constants s.t., $\frac{p+1}{2}+\tau\frac{(p+2)(p+1)}{2p}-C_9<0$. Similarly, $C_9^{'}>0$ is taken as another appropriate constant s.t., $\frac{p}{2}+\frac{p\tau}{2}-C_9^{'}<0$. This is possible when $Nh^p$ is large enough. 
    
\end{remark}

We have the transformed variables $\hat{U}^{(Y_{n+1})}_t= \hat{F}^{(Y_{n+1})}(Y_t|\bm{X}_{t-1})$ for $t = p+1, \ldots, n+1,$ and $\hat{V}^{(Y_{n+1})}_t = |\hat{U}^{(Y_{n+1})}_t -1/2|.$ On the other hand, we define the oracle variables $U_t = F(Y_t|\bm{X}_{t-1})$ and $V_t = |U_t-1/2|$, which are Uniform$(0,1)$ and Uniform$(0,1/2)$ random variables, respectively. Now we provide the following theorem for the marginal coverage guarantee of DCP-type prediction intervals with Markov$(p)$ data, given the nominal level $1-\alpha.$ 

\begin{theorem}[Finite sample unconditional validity]
\label{thm: marginal}
    Under A1 to A5, given $\alpha \in (0, 1)$, taking $5C_\delta = O\left( \frac{\log N}{\sqrt{Nh^{p}}} +h^2\right)$, we have 
    $$
    |\mathbb{P}(Y_{n+1}\in \hat{C}^\text{MDCP}_{1-\alpha}(\bm{X}_n)) - (1-\alpha)| \leq 24C_\delta + 4\exp(-8NC_\delta^2) + 2S_N;
    $$
where $S_N$ is the appropriate sequence defined in Lemma \ref{lem: mse of uhat} and converges to 0 as $n\to \infty$. By our assumption, $C_\delta$ is $o(1)$.
\end{theorem}

\begin{remark}
   The bounds for the difference in Theorem \ref{thm: marginal} consist of three terms: the approximation error $C_\delta$ in the non-conformity score distribution, and the approximation and stochastic errors between the ground truth and the estimated distribution function. The $C_\delta$ is from the proof of Lemma \ref{lem: mse of uhat}, which tends to zero as $N$ tends to infinity. The second term is the exponential decay term, which becomes a polynomial decay with log-factors. The third term is related to the high probability condition in Lemma \ref{lem: mse of uhat} defined in Remark \ref{Remark:Sn}.
\end{remark}

We should mention that the unconditional and conditional coverage rates are asymptotically equivalent for the DCP method in the regression case when the data are i.i.d after manipulating the prediction algorithm slightly. However, at first glance, it is not straightforward to extend the unconditional coverage result to the conditional case for the dependent data; see Remark \ref{Remark:DCPdiscuss} for a discussion.

\begin{remark}\label{Remark:DCPdiscuss}
    We first notice that $U_t=F\left(Y_t | \bm{X}_{t-1}\right)$ for $t = p+1, \ldots, n+1$ is independent of $\bm{X}_{t-1}$. Since $V_t = |U_t -1/2|$, $V_{n+1}$ is also independent of $\bm{X}_n$ and
$$
\mathbb{P}\left(G\left(V_{n+1}\right) \leq \alpha \mid \bm{X}_{n}\right)=\mathbb{P}\left(G\left(V_{n+1}\right) \leq \alpha\right).
$$
Since $G(\cdot)$ is the distribution function of $V_{n+1}$ and is a continuous function, we have that $\mathbb{P}\left(G\left(V_{n+1}\right) \leq \alpha\right)=\alpha$. 
Similar to the decomposition \eqref{eq: decom} in the proof of Lemma \ref{lem: mse of uhat}, we can define a piecewise function $\tilde{G}$ generated from ground-truth transformed random variable $V_t$,
$\tilde{G}(v) = \frac{1}{N}\sum_{t=p+1}^{n+1}\mathbbm{1}\{V_t <v\},$ we have $  | \mathbb{P}\left(Y_{n+1} \in \widehat{C}_{1-\alpha}^{\text {MDCP }}\left(\bm{X}_{n}\right) \mid \bm{X}_{n}\right) - 1-\alpha|$ as the error of the empirical conditional coverage rate of our MDCP prediction interval compared to the nominal level. If the independent data is considered here, we can take the leave-one-out idea to make transformation without $\bm{X}_n$, so $\mathbb{P}\left(Y_{n+1} \in \widehat{C}_{1-\alpha}^{\text {MDCP }}\left(\bm{X}_{n}\right) \mid \bm{X}_{n}\right) = \mathbb{P}\left(Y_{n+1} \in \widehat{C}_{1-\alpha}^{\text {MDCP }}\left(\bm{X}_{n}\right)\right)$. However, for the dependent data, the dependence can not be dampened by removing only one observation to do the transformation. 
\end{remark}

In the recent literature about conformal prediction, the theoretical guarantee of the conditional coverage rate for dependent data is developed in alignment with the split conformal prediction idea. \cite{chernozhukov2021distributional} first introduced and studied split distributional conformal prediction for exchangeable data. \cite{zheng2024conformal} considered the Markovian Data and created separation between the training and calibration sets so that the temporal dependence in the data can be reduced; this idea is also similar to the ``Big Block – Small Block'' technique described in Paradigm 9.3.1. of \cite{politis2019time}. The block technique is also taken by \cite{oliveira2024split} to build prediction intervals for $\beta$-mixing dataset. \cite{xu2021conformal} studied conformal inference for stationary and ergodic processes by separating regression data into two parts for nonparametric estimation and obtaining prediction residuals for the prediction interval.

Rather than relying on the blocking technique, we observe that the conditional prediction interval from MDCP is asymptotically valid, leveraging the definition of $\beta$-mixing. We summarize this interesting finding in Proposition \ref{Proposition:asymptoticMDCP}.

\begin{proposition}[The asymptotic validity of prediction interval based on MDCP ]\label{Proposition:asymptoticMDCP}
Under the assumptions required by Theorem \ref{thm: marginal}, the MDCP prediction interval is asymptotically valid, i.e., 
    $$
    |\mathbb{P}(Y_{n+1}\in \hat{C}^\text{MDCP}_{1-\alpha}(\bm{X}_n)|\bm{X_n}) - (1-\alpha)| \overset{p}{\to} 0 ~\text{for any}~\bm{X_n}.
    $$
\end{proposition}

Next, based on another dependence measure rather than the mixing property, we apply a new technique to prove that the conditional MDCP coverage rate is also asymptotically valid. To achieve this, we add one more assumption A$1^{\prime}$ as follows.
\begin{itemize}
    \item [A$1^{\prime}$] $\{Y_t\}$ forms an $L^2$-$m$-approximable series, which is defined in \cite{hormann2010weakly}; also the Markov order $p$ is less than $m$.
\end{itemize}
where the $L^q$-$m$-approximable series implies the representation,

$$
Y_t=f\left(\varepsilon_t, \varepsilon_{t-1}, \ldots\right),~\text{for each}~t;
$$
where the $\varepsilon_i$ are i.i.d. elements taking values in a real measurable space $S$ in the context of this paper, and $f$ is a measurable function $f: S^{\infty} \rightarrow S$. Moreover, we take an independent copy $\left\{\varepsilon_k^{(t)}\right\}$ of $\left\{\varepsilon_k\right\}$ for each $t$. Then, we can consider the so-called $L^q$-$m$-approximator: $
Y_t^{(m)}=f\left(\varepsilon_t, \varepsilon_{t-1}, \ldots, \varepsilon_{t-m+1}, \varepsilon_{t-m}^{(t)}, \varepsilon_{t-m-1}^{(t)}, \ldots\right)$, which satisfies several following properties:
\begin{itemize}
    \item The series $\{Y_t^{(m)}\}$ are strictly stationary and $m$-dependent;
    \item $Y_t^{(m)}$ is equal in distribution to $Y_t$ for each $t$;
    \item $\sum_{m=1}^{\infty}\left( \mathbb{E}|Y_m-Y_m^{(m)}|^q\right)^{1/q}<\infty$.
\end{itemize}

Before going to the statement of our theorem, Remark \ref{Remark:differentmeasure} gives a short introduction about different types of conditions to measure the dependence. 

\begin{remark}[Different dependence measures]\label{Remark:differentmeasure}
The first and maybe most common approach to measure the dependency is mixing property, e.g., $\alpha$- and $\beta$-mixing; see Section \ref{Appendix:Lemma} for the definition of these two mixing properties; see the relationship between different mixing conditions from \cite{bradley2005basic}. Besides mixing conditions, the namely $\nu$-stable process has been defined by \cite{bierens1983uniform}, sharing a similar idea with Section 21 of \cite{billingsley1968convergence}. Subsequently, \cite{potscher1997dynamic} proposed a more general $L_q$-Approximable idea, which is more general than $\nu$-stability. Recently, \cite{hormann2010weakly} further generalized the previous results to obtain the $L^q$-$m$-approximable process. However, the  $L^q$-$m$-approximable property is not comparable to the mixing conditions. For example, the AR(1) process $Y_t= \rho Y_{t-1}+\varepsilon_t; 0 \leq\rho \leq 1/2,$ with independent Bernoulli innovations, does not have the $\alpha$-mixing property; see \cite{andrews1984non} for the proof. In addition, Theorem 17.3.3. of \cite{ibragimov1975independent} gave a condition s.t. the Gaussian process is strongly mixing. In other words, it is likely that an AR model with normal innovation is still not strongly mixing. However, the AR(1) process with independent Bernoulli innovations is $L^q$-$m$-approximable obviously. Thus, $L^q$-$m$-approximable property does not imply the mixing conditions. On the other hand, the $L^q$-$m$-approximable property is acquired by restricting the dependent data that possesses a specific data-generating structure. Thus, the $L^q$-$m$-approximable and mixing conditions are not comparable directly.
\end{remark}

To show the asymptotically conditional validity of MDCP with $L^q$-$m$-approximable property, we still need to quantify the ``dependence'' of the series with the following assumption:

\begin{itemize}
    \item [A6] For the series $\{Y_t\}$, we denote $\sum_{m=1}^{\infty}\left( \mathbb{E}|Y_m-Y_m^{(m)}|^q\right)^{1/q} = \sum_{m=1}^{\infty}\delta(m)<\infty$. We require that $\delta(m)$ converges to 0 fast as $m\to\infty$ s.t., $\delta(m) = o(h^{p+1})$ and $m = o(n)$.
\end{itemize}

In this paper, we consider $L^2$-$m$-approximable series, and the following theorem verifies the asymptotic validity of the conditional prediction interval from MDCP. 
\begin{theorem}[Asymptotically conditional validity]
\label{thm: conditional}
    Under A$1^{\prime}$ and A2 to A6, given $\alpha \in (0, 1)$, taking $\{Y_t^{(m)}\}$ as the $L^2$-$m$-approximator of the original series with the latest $p$ values replaced with the values from the original series, we have 
    $$
    |\mathbb{P}(Y_{n+1}\in \hat{C}^\text{MDCP}_{1-\alpha}(\bm{X}_n)|\bm{X_n}) - (1-\alpha)| \overset{p}{\to} 0.
    $$
\end{theorem}

\begin{remark}[Discussion of $\{Y_t^{(m)}\}$ series]
    To simplify the proof, the latest $p$ values of $\{Y_t^{(m)}\}$ are replaced by the values from the original series, so that the conditional prediction on two series is given the same latest information, which provides the convenience to the proof. This replacement does change the property of $\{Y_t^{(m)}\}$. It will still be $m$-dependent and strictly stationary.
\end{remark}

Finally, we wrap up this section with a proposition.

\begin{proposition}[The asymptotic equivalence between MDCP and PMDCP]

We claim that all the above theoretical results will be held with the PMDCP technique since the leave-one-out CDF estimator converges to the standard one in a deterministic way, assuming the domain of the series is compact; see more discussion of the leave-one-out non-parametric estimator from \cite{pan2016bootstrap,politis2023multi}.

\end{proposition}

\begin{remark}
    In general, it is impossible to obtain the conditional coverage except for some specific structure assumptions \citep{vovk2012conditional,foygel2021limits}. \cite{chernozhukov2021distributional} bridged the gap between the unknown conditional coverage guarantee and non-exchangeable data, even when the CDF estimation is misspecified.
    \cite{gibbs2025conformal} reformulated the conditional coverage as coverage over a class of covariate shifts and obtained an explicit form of finite-sample conditional coverage. 
    \cite{xu2021conformal} established the conditional coverage for time series with error following different processes.
    In our work, we take the definition of $L^q$-$m$-approximable series to prove that the conditional and unconditional coverage probabilities are asymptotically equivalent. And the $L^q$-$m$-approximator can adapt the result for some non-mixing processes. 
For example, for AR(1) model with i.i.d. Bernoulli innovation, it could be a non-mixing process, but it is an $m$-approximable series. 
In short, our proof covers mixing Markov processes, and it also covers the non-mixing but $m$-approximable process. 
\end{remark}

\section{Numerical Experiments}
\label{sec: simu}
\subsection{Simulation studies}
\label{sec: simu_1}
In this section, we use Monte Carlo simulations to evaluate the performance of the conformal prediction methods proposed in this paper through average coverage level (CVR) and length (LEN). For each type of prediction method, we generate $R = 1000$ datasets, each with size $n$. To evaluate the performance of the different prediction methods, the following two models were chosen in order to generate Markov process data (of order $p = 1$).
\begin{itemize}
    \item[] Model 1:$Y_{t+1}=\sin(Y_t)+e_{t+1}$~;~Model 2:$Y_{t+1}=0.8\log (3Y_t^2+1)+e_{t+1}$.
\end{itemize}
In the above, the errors $\{e_t\}$ are chosen either from $i.i.d.$ $N(0, 1)$ distribution or $i.i.d.$ Laplace distribution rescaled to variance 1. We chose sample sizes $n = 50, 100, 250, 500, 1000$, and constructed 95\% and 90\% prediction intervals. We choose $B=250$ to compute the prediction interval of MF-based methods. We used the standard normal cumulative distribution function restricted on $[-2,2]$ as $K$ and chose $w(x) = \frac{1}{\sqrt{2\pi}}e^{-\frac{x^2}{2}}$ the standard normal density. We chose bandwidths $h$ and $h_0$ for kernel estimators by {\tt npcdistbw} function from {\tt R} package {\tt np}. This {\tt npcdistbw} function can accommodate multivariate data pairs $(Y_{t}, X_{t-1}, \ldots, X_{t-p})$ under $p >1.$ 

For each generated dataset, we use the oracle model to generate $S$ number of pseudo values of $(Y_{n+1,j}, j = 1, 2, \ldots, S)$. We take $S = 5000$ throughout our simulation studies. The number $S$ stands for how many future pseudo values of $Y_{n+1,j}$ we use to evaluate the empirical coverage rate of different Prediction Intervals (PIs) $(L_i, U_i), i = 1, 2, \ldots, R $. Then, the average coverage rate and prediction interval length over $R$ replications can be computed as follows:
\begin{equation}\label{Eq:CVRcriterion}
    CVR = \frac{1}{R}\sum_{i=1}^{R}CVR_i \text{  and  }LEN=\frac{1}{R}\sum_{i=1}^{R}LEN_i;
\end{equation}
where 
\begin{equation}\label{Eq:CVRi}
CVR_i =\frac{1}{S}\sum_{j=1}^S \mathbbm{1}\{Y_{n+1,j}\in[L_i,U_i]\} \text{  and  }LEN_i = U_i-L_i.
\end{equation}
For the $i$-th dataset of size $n$ (where $i = 1, 2, \ldots , R$), we have the prediction interval $(L_i, U_i)$ of $Y_{n+1}$ given $X_n = x_{n,i}$, where $x_{n,i}$ is the last observation from the $i$-th dataset. The values of these $\{x_{n,i}\}_{i=1}^R$ are different for each $i$; therefore, the above $CVR$ in (\ref{Eq:CVRcriterion}) is an estimate of the unconditional coverage level. Besides, we also compute the sample standard deviation of 1000 $CVR_i$ and $LEN_i$ to measure the variability of different PIs according to the coverage level and length. The $CVR_i$ from an ideal PI should be concentrated around the nominal level, and the corresponding interval length should be as small as possible. All these results across different simulation models and errors are presented in Tables \ref{Table:M1normal}-\ref{Table:M2lap}.

\begin{table}[htbp]
\center
\caption{Model 1: $Y_{t+1} =\sin(Y_t)+e_{t+1}$ with normal innovations.}
\begin{tabular}{l|cccc|cccc}
  \hline
 &\multicolumn{4}{c|}{Nominal coverage 90\%} &\multicolumn{4}{c}{Nominal coverage 95\%} \\
\hline
$n=50$ & CVR & LEN & CVR Sd & LEN Sd & CVR & LEN & CVR Sd & LEN Sd \\
MDCP &  0.892 & 3.619 & 0.066 & 0.899 & 0.951 & 4.657 & 0.042 & 1.017\\
PMDCP & 0.896 & 3.620 & 0.073 & 0.668 & 0.955 & 4.715 & 0.048 & 0.978\\
MF & 0.860 & 3.220 & 0.078 & 0.574 & 0.909 & 3.742 & 0.065 & 0.652 \\
PMF & 0.911 & 4.087 & 0.067 & 0.587 & 0.955 & 4.991 & 0.048 & 0.879\\
\hline
$n=100$  &\multicolumn{4}{c|}{} &\multicolumn{4}{c}{}\\
MDCP &  0.896 & 3.560 & 0.054 & 0.870 & 0.948 & 4.336 & 0.035 & 0.951\\
PMDCP &  0.894 & 3.479 & 0.074 & 0.545 & 0.945 & 4.244 & 0.059 & 0.677\\
MF & 0.866 & 3.210 & 0.069 & 0.469 & 0.914 & 3.725 & 0.058 & 0.537\\
PMF & 0.900 & 4.020 & 0.080 & 0.408 & 0.950 & 4.860 & 0.052 & 0.552 \\
\hline
$n=250$ &\multicolumn{4}{c|}{} &\multicolumn{4}{c}{} \\
MDCP &  0.894 & 3.410 & 0.040 & 0.746 & 0.948 & 4.156 & 0.026 & 0.822\\
PMDCP &0.895 & 3.349 & 0.042 & 0.351 & 0.950 & 4.104 & 0.032 & 0.451\\
MF & 0.876 & 3.202 & 0.046 & 0.386 & 0.926 & 3.748 & 0.037 & 0.448 \\
PMF & 0.900 & 3.930 & 0.074 & 0.312 & 0.949 & 4.717 & 0.054 & 0.404\\
\hline
$n=500$  &\multicolumn{4}{c|}{} &\multicolumn{4}{c}{}\\
MDCP &  0.894 & 3.367 & 0.032 & 0.711 & 0.946 & 4.080 & 0.023 & 0.834\\
PMDCP &0.892 & 3.287 & 0.044 & 0.310 & 0.945 & 3.971 & 0.035 & 0.384 \\
MF &0.879 & 3.196 & 0.040 & 0.308 & 0.929 & 3.749 & 0.033 & 0.363  \\
PMF & 0.899 & 3.935 & 0.079 & 0.260 & 0.949 & 4.699 & 0.055 & 0.340  \\
\hline
$n=1000$ &\multicolumn{4}{c|}{} &\multicolumn{4}{c}{} \\
MDCP &  0.892 & 3.283 & 0.026 & 0.472 & 0.945 & 3.989 & 0.018 & 0.681\\
PMDCP &0.891 & 3.254 & 0.035 & 0.247 & 0.945 & 3.924 & 0.029 & 0.309\\
MF &0.885 & 3.225 & 0.036 & 0.287 & 0.934 & 3.794 & 0.030 & 0.350 \\
PMF & 0.897 & 3.931 & 0.076 & 0.235 & 0.947 & 4.695 & 0.052 & 0.308   \\
\hline
\end{tabular}
\label{Table:M1normal}
\end{table}

\begin{table}[htbp]
\center
\caption{Model 1: $Y_{t+1} =\sin(Y_t)+e_{t+1}$ with Laplace innovations.}
\begin{tabular}{l|cccc|cccc}
  \hline
 &\multicolumn{4}{c|}{Nominal coverage 90\%} &\multicolumn{4}{c}{Nominal coverage 95\%} \\
\hline
$n=50$ & CVR & LEN & CVR Sd & LEN Sd & CVR & LEN & CVR Sd & LEN Sd \\
MDCP &  0.895 & 3.926 & 0.068 & 1.588 & 0.954 & 5.500 & 0.042 & 1.838\\
PMDCP & 0.891 & 3.953 & 0.109 & 1.305 & 0.953 & 5.694 & 0.074 & 1.805\\
MF & 0.831 & 3.015 & 0.116 & 0.833 & 0.880 & 3.667 & 0.106 & 1.028 \\
PMF &0.923 & 4.477 & 0.058 & 1.117 & 0.962 & 5.907 & 0.035 & 1.646  \\
\hline
$n=100$  &\multicolumn{4}{c|}{} &\multicolumn{4}{c}{}\\
MDCP &  0.894 & 3.815 & 0.057 & 1.681 & 0.948 & 5.121 & 0.038 & 1.884\\
PMDCP &  0.897 & 3.757 & 0.089 & 1.018 & 0.948 & 5.157 & 0.076 & 1.415\\
MF & 0.846 & 3.003 & 0.084 & 0.709 & 0.895 & 3.706 & 0.075 & 0.913\\
PMF &  0.919 & 4.315 & 0.066 & 0.758 & 0.962 & 5.745 & 0.042 & 1.219 \\
\hline
$n=250$ &\multicolumn{4}{c|}{} &\multicolumn{4}{c}{} \\
MDCP &  0.890 & 3.486 & 0.045 & 1.314 & 0.947 & 4.795 & 0.031 & 1.635\\
PMDCP &0.899 & 3.502 & 0.053 & 0.677 & 0.954 & 4.974 & 0.040 & 1.128\\
MF & 0.865 & 3.047 & 0.052 & 0.562 & 0.915 & 3.811 & 0.042 & 0.728\\
PMF &  0.912 & 4.085 & 0.066 & 0.472 & 0.959 & 5.382 & 0.038 & 0.788 \\
\hline
$n=500$  &\multicolumn{4}{c|}{} &\multicolumn{4}{c}{}\\
MDCP & 0.891 & 3.472 & 0.038 & 1.368 & 0.945 & 4.647 & 0.027 & 1.652\\
PMDCP &0.895 & 3.388 & 0.053 & 0.535 & 0.949 & 4.664 & 0.047 & 0.886 \\
MF &0.871 & 3.087 & 0.056 & 0.487 & 0.920 & 3.880 & 0.048 & 0.630\\
PMF & 0.912 & 3.997 & 0.056 & 0.379 & 0.959 & 5.181 & 0.030 & 0.603 \\
\hline
$n=1000$ &\multicolumn{4}{c|}{} &\multicolumn{4}{c}{} \\
MDCP & 0.893 & 3.470 & 0.033 & 1.462 & 0.947 & 4.643 & 0.022 & 1.786\\
PMDCP &0.894 & 3.332 & 0.050 & 0.471 & 0.947 & 4.480 & 0.045 & 0.703\\
MF &0.879 & 3.145 & 0.041 & 0.476 & 0.928 & 3.976 & 0.034 & 0.620  \\
PMF & 0.906 & 3.940 & 0.069 & 0.337 & 0.954 & 5.075 & 0.040 & 0.552 \\
\hline
\end{tabular}
\label{Table:M1lap}
\end{table}

\begin{table}[htbp]
\center
\caption{Model 2: $Y_{t+1}=0.8\log (3Y_t^2+1)+e_{t+1}$ with normal innovations.}
\begin{tabular}{l|cccc|cccc}
  \hline
 &\multicolumn{4}{c|}{Nominal coverage 90\%} &\multicolumn{4}{c}{Nominal coverage 95\%} \\
\hline
$n=50$ & CVR & LEN & CVR Sd & LEN Sd & CVR & LEN & CVR Sd & LEN Sd \\
MDCP &  0.881 & 3.822 & 0.074 & 1.624 & 0.945 & 5.358 & 0.053 & 2.200\\
PMDCP &  0.886 & 3.607 & 0.087 & 0.705 & 0.947 & 4.822 & 0.061 & 1.141\\
MF & 0.852 & 3.234 & 0.091 & 0.639 & 0.903 & 3.765 & 0.076 & 0.733\\
PMF & 0.922 & 4.023 & 0.065 & 0.797 & 0.962 & 5.066 & 0.046 & 1.396 \\
\hline
$n=100$  &\multicolumn{4}{c|}{} &\multicolumn{4}{c}{}\\
MDCP &  0.889 & 3.720 & 0.059 & 1.562 & 0.942 & 4.808 & 0.042 & 2.076\\
PMDCP & 0.889 & 3.473 & 0.072 & 0.563 & 0.943 & 4.286 & 0.055 & 0.741 \\
MF & 0.864 & 3.218 & 0.070 & 0.505 & 0.914 & 3.748 & 0.057 & 0.578\\
PMF & 0.914 & 3.748 & 0.057 & 0.585 & 0.957 & 4.580 & 0.041 & 0.818  \\
\hline
$n=250$ &\multicolumn{4}{c|}{} &\multicolumn{4}{c}{} \\
MDCP &  0.887 & 3.466 & 0.042 & 1.240 & 0.945 & 4.360 & 0.029 & 1.593\\
PMDCP &0.887 & 3.319 & 0.056 & 0.396 & 0.945 & 4.100 & 0.048 & 0.520\\
MF & 0.873 & 3.197 & 0.055 & 0.397 & 0.923 & 3.742 & 0.046 & 0.455\\
PMF & 0.900 & 3.466 & 0.051 & 0.435 & 0.948 & 4.192 & 0.042 & 0.579\\
\hline
$n=500$  &\multicolumn{4}{c|}{} &\multicolumn{4}{c}{}\\
MDCP &  0.888 & 3.385 & 0.036 & 1.102 & 0.943 & 4.213 & 0.025 & 1.553\\
PMDCP &0.887 & 3.266 & 0.051 & 0.344 & 0.941 & 3.955 & 0.045 & 0.433  \\
MF &0.877 & 3.193 & 0.047 & 0.337 & 0.927 & 3.753 & 0.039 & 0.399 \\
PMF & 0.895 & 3.366 & 0.047 & 0.370 & 0.944 & 4.043 & 0.038 & 0.471 \\
\hline
$n=1000$ &\multicolumn{4}{c|}{} &\multicolumn{4}{c}{} \\
MDCP &  0.886 & 3.290 & 0.028 & 0.892 & 0.942 & 4.041 & 0.019 & 1.248\\
PMDCP &0.885 & 3.209 & 0.045 & 0.276 & 0.941 & 3.878 & 0.040 & 0.339\\
MF &0.883 & 3.208 & 0.042 & 0.302 & 0.932 & 3.771 & 0.036 & 0.351\\
PMF &  0.894 & 3.323 & 0.043 & 0.323 & 0.943 & 3.970 & 0.036 & 0.405\\
\hline
\end{tabular}
\label{Table:M2normal}
\end{table}

\begin{table}[htbp]
\center
\caption{Model 2: $Y_{t+1}=0.8\log (3Y_t^2+1)+e_{t+1}$ with Laplace innovations.}
\begin{tabular}{l|cccc|cccc}
  \hline
  &\multicolumn{4}{c|}{Nominal coverage 90\%} &\multicolumn{4}{c}{Nominal coverage 95\%} \\
\hline
$n=50$ & CVR & LEN & CVR Sd & LEN Sd & CVR & LEN & CVR Sd & LEN Sd \\
MDCP &  0.889 & 4.329 & 0.073 & 2.447 & 0.949 & 6.402 & 0.049 & 2.850\\
PMDCP &0.890 & 3.981 & 0.108 & 1.304 & 0.949 & 5.753 & 0.070 & 1.829\\
MF & 0.832 & 3.055 & 0.106 & 0.877 & 0.881 & 3.678 & 0.092 & 1.055 \\
PMF & 0.923 & 4.466 & 0.071 & 1.261 & 0.960 & 5.991 & 0.044 & 2.022 \\
\hline
$n=100$  &\multicolumn{4}{c|}{} &\multicolumn{4}{c}{}\\
MDCP &   0.889 & 4.142 & 0.063 & 2.525 & 0.944 & 5.745 & 0.042 & 2.923\\
PMDCP & 0.887 & 3.757 & 0.110 & 1.114 & 0.940 & 5.125 & 0.094 & 1.531 \\
MF &  0.834 & 2.995 & 0.106 & 0.762 & 0.883 & 3.654 & 0.097 & 0.942 \\
PMF & 0.908 & 4.059 & 0.083 & 1.059 & 0.950 & 5.358 & 0.065 & 1.514\\
\hline
$n=250$ &\multicolumn{4}{c|}{} &\multicolumn{4}{c}{} \\
MDCP &   0.888 & 3.637 & 0.047 & 1.941 & 0.946 & 5.214 & 0.032 & 2.602\\
PMDCP &0.895 & 3.525 & 0.068 & 0.785 & 0.951 & 5.004 & 0.053 & 1.221\\
MF & 0.860 & 3.038 & 0.064 & 0.638 & 0.910 & 3.812 & 0.057 & 0.849  \\
PMF & 0.902 & 3.681 & 0.059 & 0.824 & 0.950 & 4.942 & 0.042 & 1.167 \\
\hline
$n=500$  &\multicolumn{4}{c|}{} &\multicolumn{4}{c}{}\\
MDCP & 0.891 & 3.621 & 0.041 & 1.969 & 0.946 & 5.038 & 0.027 & 2.630\\
PMDCP &0.893 & 3.417 & 0.063 & 0.625 & 0.947 & 4.738 & 0.055 & 1.006 \\
MF & 0.870 & 3.108 & 0.055 & 0.554 & 0.919 & 3.912 & 0.047 & 0.733\\
PMF & 0.900 & 3.567 & 0.051 & 0.663 & 0.948 & 4.798 & 0.040 & 1.055 \\
\hline
$n=1000$ &\multicolumn{4}{c|}{} &\multicolumn{4}{c}{} \\
MDCP & 0.891 & 3.526 & 0.035 & 1.959 & 0.945 & 4.782 & 0.024 & 2.394\\
PMDCP & 0.893 & 3.327 & 0.052 & 0.534 & 0.946 & 4.501 & 0.045 & 0.771\\
MF & 0.877 & 3.134 & 0.043 & 0.517 & 0.926 & 3.955 & 0.035 & 0.648\\
PMF &  0.899 & 3.455 & 0.044 & 0.593 & 0.947 & 4.580 & 0.033 & 0.855 \\
\hline
\end{tabular}
\label{Table:M2lap}
\end{table}

To check the conditional coverage performance of all PI candidates, we consider the two situations: $n = 50$ and $n = 250$. We present the histogram of $CVR_i$ for each $95\%$-PI with Model-1 and Model-2 data generated from normal in Figure \ref{fig:M1normal} and Figure \ref{fig:M2normal}, respectively. Since the pattern is similar to that with Laplace errors, we show the graphs about the performance of different PIs according to the coverage rate and interval length with data generated by Model-1 and Model-2 with Laplace errors in Appendix \ref{Appendix:additionalgraphs}.


As we see from the Figures \ref{fig:M1normal}-\ref{fig:M2normal}, MDCP and PMDCP methods will ``overcorrect'' some prediction intervals to reach a conditional coverage rate beyond the nominal levels, especially for the small sample size. This phenomenon also happens for MF and PMF PIs. When the sample size is large, the histograms of $CVR_i$ of MDCP and PMDCP are more concentrated around the assigned nominal level. If we take a closer look, the length of PI from MDCP method could be pretty large even when the sample size is large, which is not desirable in practice. This phenomenon also happens in the simulation results with Laplace errors presented in Appendix \ref{Appendix:additionalgraphs}. In addition, the histogram of $CVR_i$ of PMF reveals more high conditional coverage cases compared to $CVR_i$ of PMDCP, but results in larger PI lengths as a sacrifice. In practice, the overcovering PI may be better than the undercovering PI even with a slightly larger length.

On the other hand, from the results shown in Tables \ref{Table:M1normal}-\ref{Table:M2lap}, the estimated unconditional coverage rates of MDCP, PMDCP and PMF PIs are all close to the nominal level even for a small sample size. Only the PI of MF shows the undercoverage issue. Among MDCP, PMDCP and PMF methods, the MDCP has a larger prediction interval length and a corresponding larger sample standard deviation, which implies it is more likely to have a super wide PI, consistent with the finding from \ref{fig:M1normal}-\ref{fig:M2normal}. Again, PMF PI tends to have a slightly higher coverage rate than PMDCP, but possesses a larger PI length as a sacrifice. All PIs have better performance as the sample size increases. It also validates the asymptotic theory of PIs in Section \ref{sec: asym}.
\begin{figure}[!ht]
    \centering
    \includegraphics[scale=0.9]{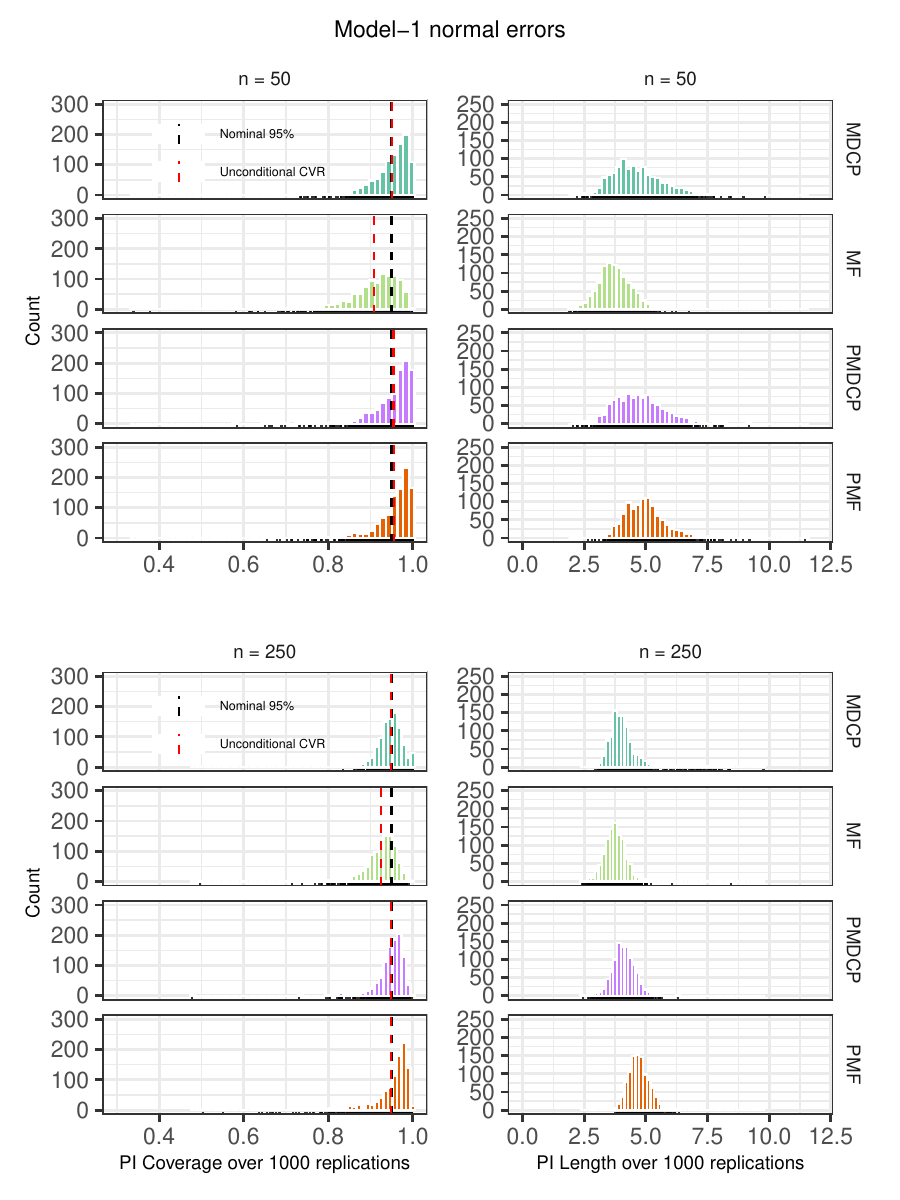}
    \caption{The conditional coverage rate and interval length for different 95$\%$-PIs with simulated data from Model 1 with normal error.}
    \label{fig:M1normal}
\end{figure}
\begin{figure}[!ht]
    \centering
    \includegraphics[scale=0.9]{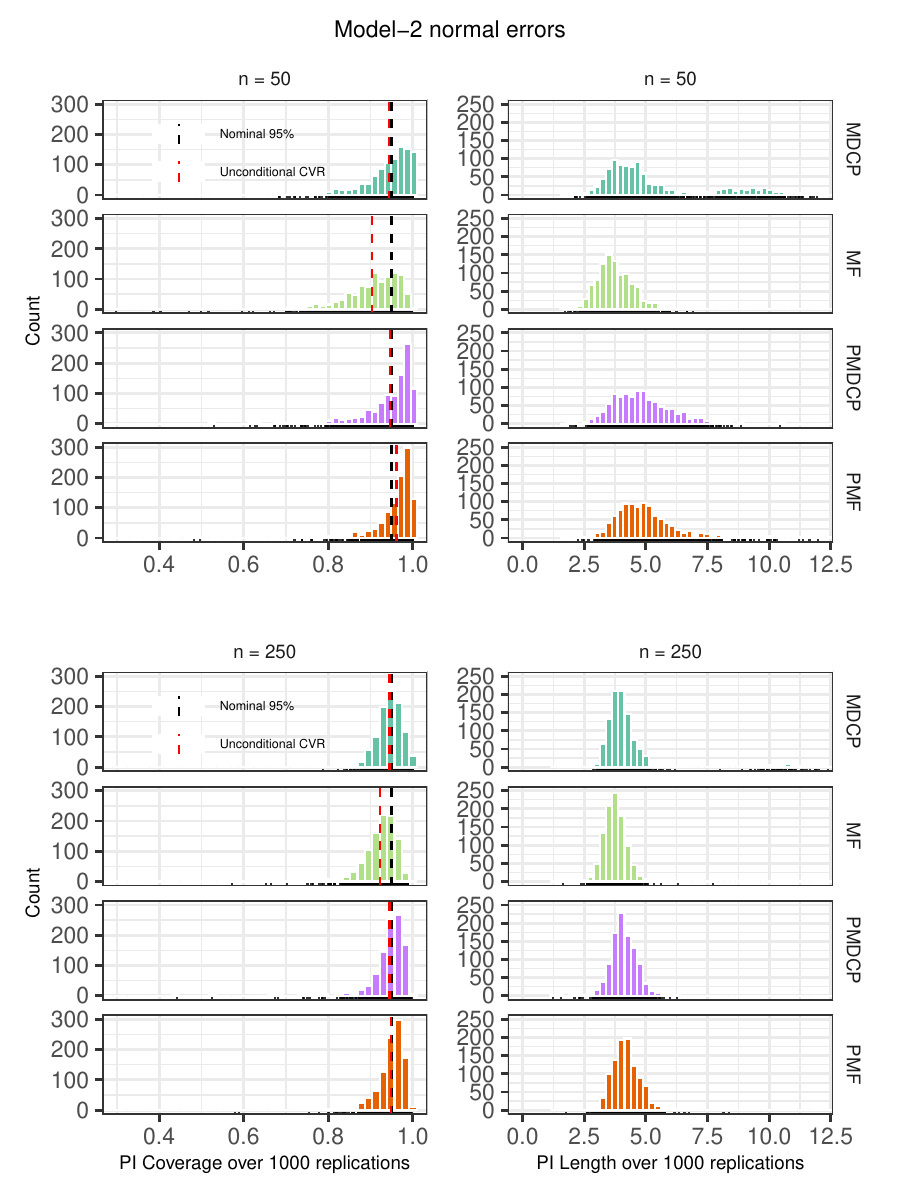}
    \caption{The conditional coverage rate and interval length for different 95$\%$-PIs with simulated data from Model 2 with normal error.}
    \label{fig:M2normal}
\end{figure}
From the results implied by both tables and figures, we make our suggestions for users to choose a method in practice:
\begin{itemize}
    \item If the computational resource is limited, we recommend the MF method since it can achieve a great unconditional coverage rate when the sample size is large enough. The MDCP method could also be a candidate since it gives a satisfactory unconditional CVR, but it brings the risk of obtaining a super wide PI; see the top right subgraph of Figure \ref{fig:M2normal}, where the histogram of MDCP PI lengths has a fat tail, which indicates the super wide PIs. A similar phenomenon is also observed from Appendix \ref{Appendix:additionalgraphs} and real data studies in the following section.
    \item If the computational resource is not an issue and users prefer the overcoverage rather than the undercoverage, we recommend the PMF method since it possesses a good unconditional coverage rate and tends to give more conditional overcoverage cases. 
    \item If the computational resource is not an issue and users prefer a balanced undercoverage and overcoverage for conditional forecasting, we recommend the PMDCP method since it possesses a good unconditional coverage rate and the histograms of its $CVR_i$ are more concentrated around the nominal level. 
\end{itemize}

\subsection{Real data studies: SP 500 index }
Beyond the simulations, we take the log returns of the weekly S$\&$P 500 price index from 1 January 1988 to 31 December 1997, i.e., in total 521 observations, to compare different methods on building prediction intervals for real data\footnote{The data can be downloaded from investing.com.}. Revealed by \cite{chen2012testing}, this returns series can be seen as a Markov(1) data. To evaluate the performance of different PIs, we apply the rolling window prediction idea with a window size $w$ of different quarters. For example, given the window size $w=90$ and the order $p$ of Markov process, we use the 90 previous observation pairs $\mathcal{D}_{train} =\left\{\left(\bm{X}_{t-90}, Y_{t-89}\right), \ldots,\left(\bm{X}_{t-1}, Y_{t}\right)\right\}$ to provide a prediction interval of $Y_{t+1}$ given $\bm{X}_t = (Y_{t-p+1}, \ldots, Y_t)$. Repeating this prediction procedure throughout the whole series, we can define the empirical rolling average coverage by $
CVR = \frac{1}{n-w}\sum_{t = w+1}^n \mathbbm{1}\{Y_{t} \in [L_t, U_t]\};$ where $L_t$ and $U_t$ are different lower and upper bounds in the prediction interval obtained from these models, respectively, for each rolling prediction step. In addition, the average length of the prediction interval can be defined as $
LEN = \frac{1}{n-w}\sum_{t = w+1}^n(U_t - V_t)$. Again, we compute ``LEN Sd'', which represents the sample deviation of all rolling PI lengths, to measure the spread of the length of each rolling PI. All results are presented in Table \ref{tab: PIsp500}. From there, we observe a similar pattern to the simulations. PMF and PMDCP are two superior methods. Although MDCP can return a decent coverage rate, it possesses the risk of getting wide PIs, especially for $w = 250$, which is implied by the largest ``LEN Sd'' among all methods. This phenomenon is revealed in detail by Figure \ref{fig: histSP_250}. Therefore, we still recommend the PMF or PMDCP methods.  

\begin{table}[!ht]
    \centering
    \begin{tabular}{l|ccc|ccc}
    \hline
   Method  &  \multicolumn{3}{c|}{Nominal coverage 90\% }& \multicolumn{3}{c}{Nominal coverage 95\% } \\
      \hline
$w = 100$    & CVR    & LEN    & LEN Sd & CVR    & LEN    & LEN Sd \\
MF    & 0.8575 & 0.0486 & 0.0115 & 0.8884 & 0.0583 & 0.0140 \\
PMF   & 0.8789 & 0.0542 & 0.0131 & 0.9382 & 0.0678 & 0.0174 \\
MDCP  & 0.8717 & 0.0520 & 0.0126 & 0.9335 & 0.0651 & 0.0168 \\
PMDCP & 0.8717 & 0.0526 & 0.0121 & 0.9335 & 0.0653 & 0.0163 \\
   \hline
 $w = 250$    &     &    &   &    &   &   \\
MF    & 0.8561 & 0.0458 & 0.0067 & 0.9188 & 0.0555 & 0.0092 \\
PMF   & 0.8708 & 0.0491 & 0.0079 & 0.9299 & 0.0621 & 0.0121 \\
MDCP  & 0.8708 & 0.0488 & 0.0136 & 0.9299 & 0.0619 & 0.0158 \\
PMDCP & 0.8635 & 0.0472 & 0.0070 & 0.9373 & 0.0610 & 0.0103 \\
\hline
    \end{tabular}
    \caption{The out-of-sample prediction performance of different PIs on log-returns of S$\&$P500 price index with 100 or 250 rolling window size.}
    \label{tab: PIsp500}
\end{table}

\begin{figure}[!ht]
    \centering
    \includegraphics[scale=0.5]{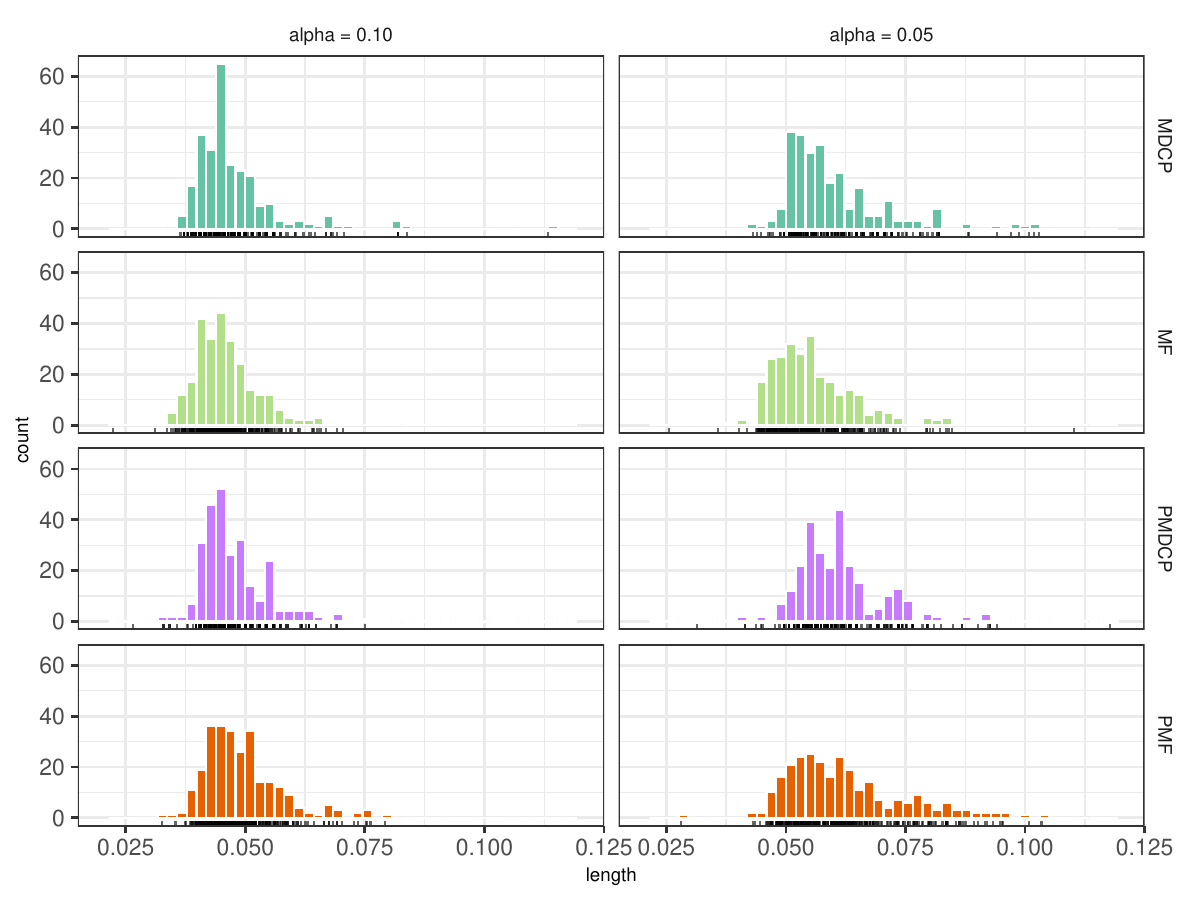}
    \caption{The histogram of length for rolling 90\%- and 95\%-PIs with log-returns of S$\&$P500 price index when $w = 250$.}
    \label{fig: histSP_250}
\end{figure}

\FloatBarrier
\section*{Disclosure Statement}
The authors have no conflicts of interest to declare.
\section*{Data Availability Statement}
Data sharing is not applicable to this article as no new data were created or analyzed in this study.

%% file: supp.tex
\section{Proofs in Section 3}
\subsection{Proof of Theorem \ref{thm: marginal}}
Define $\hat{G}(v) = 1/N\sum_{t=p+1}^{n+1}\mathbbm{1}\{\hat{V}^{(Y_{n+1})}_t<v\}$.
Notice that 
$$
\mathbb{P}(Y_{n+1}\in \hat{C}^\text{MDCP}_{1-\alpha}(\bm{X}_{n}))=\mathbb{P}\left(\frac{1}{N}\sum_{t=p+1}^{n+1}\mathbbm{1}\left\{\hat{V}_t^{(Y_{n+1})}\geq \hat{V}_{n+1}^{(Y_{n+1})}\right\}>\alpha\right) =\mathbb{P}\left(1-\hat{G}(\hat{V}_{n+1}^{(Y_{n+1})})>\alpha\right).
$$
Denote $G(\cdot)$ as the true distribution function of $V_{n+1}$. Since $G(\cdot)$ is a continuous function, $G(V_{n+1})$ follows from Uniform(0,1) by PIT and has the probability
 $
\mathbb{P}(G(V_{n+1})<1-\alpha) = 1-\alpha.
 $
 We have the decomposition
 \begin{equation}
 \label{eq: decom}
 \begin{aligned}
     |\mathbb{P}(Y_{n+1}\in \hat{C}^\text{MDCP}_{1-\alpha}(\bm{X}_n)) - (1-\alpha)| &= \left|\mathbb{P}\left( \frac{1}{N}\sum_{t=p+1}^{n+1}\mathbbm{1}\{\hat{V}_t^{(Y_{n+1})} \geq \hat{V}_{n+1}^{(Y_{n+1})} \}>\alpha\right) -(1-\alpha)\right|\\
     &= \left|\mathbb{P}\left( 1- \hat{G}(\hat{V}_{n+1}^{(Y_{n+1})})>\alpha\right) -\mathbb{P}\left( 1- G({V}_{n+1})>\alpha\right)\right|\\
     &\leq \mathbb{P}(|\hat{G}(\hat{V}^{(Y_n+1)}_{n+1})-G(V_{n+1})|>12C_\delta)\\
     &+\mathbb{P}(|G(V_{n+1})-(1-\alpha)|\leq 12C_\delta)\\
&\leq \mathbb{P}(|\hat{G}(\hat{V}^{(Y_n+1)}_{n+1})-G(V_{n+1})|>12C_\delta)+24C_\delta;
 \end{aligned}
 \end{equation}
where the inequality in \eqref{eq: decom} is due to the inequality $|\mathbb{P}(X \leq z)-\mathbb{P}(Y \leq z)| \leq \mathbb{P}(|X-Y|>\delta)+\mathbb{P}(|X-z| \leq \delta)$ for any $\delta >0$; $C_\delta \asymp \frac{\log N}{\sqrt{Nh^{p}}} +h^2$; the last equation is due to the fact that $G(V_{n+1})$ follows from Uniform(0,1).
 
Define a piecewise function $\tilde{G}$ generated from ground-truth transformed random variable $V_t$,
$$
\tilde{G}(v) = \frac{1}{N}\sum_{t=p+1}^{n+1}\mathbbm{1}\{V_t <v\}.
$$ 
Thus, we can have the decomposition for the first term on the r.h.s. of \eqref{eq: decom}
\begin{align*}
    \mathbb{P}(|\hat{G}(\hat{V}^{(Y_n+1)}_{n+1})-G(V_{n+1})|>12C_\delta) &\leq \mathbb{P}(|\hat{G}(\hat{V}^{(Y_n+1)}_{n+1}) - \tilde{G}(\hat{V}^{(Y_n+1)}_{n+1})|>8C_\delta) \\
    &\qquad+ \mathbb{P}(|\tilde{G}(\hat{V}^{(Y_n+1)}_{n+1})- G(\hat{V}^{(Y_n+1)}_{n+1})|>2C_\delta)\\
    &\qquad+\mathbb{P}(| G(\hat{V}^{(Y_n+1)}_{n+1}) - G(V_{n+1})|>2C_\delta).
\end{align*}
Define the events $E_1 :=\{\sup_v|\hat{G}(v)-\tilde{G}(v)|> 8C_\delta\}$, $E_2:=\{\sup_v|\tilde{G}(v)-G(v)|>2C_\delta\}$ and $E_3 = \{| G(\hat{V}^{(Y_n+1)}_{n+1}) - G(V_{n+1})|>2C_\delta$\}.\\
\noindent{\sc Step 1. Bound $E_3$.} Since $G$ is the ground true distribution of $V$ which is uniform[0,1/2], the probability of event $E_3$ happening is equivalent to $\{| \hat{V}^{(Y_n+1)}_{n+1} - V_{n+1}|>C_\delta\}$ happening. Moreover, $V$ has similar transformations to $\hat{V}$, i.e., $V_t = |U_t - 1/2|$, so 
\begin{align*}
    | \hat{V}^{(Y_n+1)}_{n+1} - V_{n+1}|  = \bigg||\hat{U}_{n+1}^{(Y_{n+1})}-1/2|-|{U}_{n+1}-1/2|\Bigg| \leq |\hat{U}_{n+1}^{(Y_{n+1})}-{U}_{n+1}|.
\end{align*}
Lemma \ref{lem: mse of uhat} provides $|\hat{U}_{n+1}^{(Y_{n+1})}-{U}_{n+1}| \leq C_\delta)$ with $C_\delta = O( \frac{\log N}{\sqrt{Nh^{p}}} +h^2)$ which is defined in Lemma \ref{lem: mse of uhat} with probability at least $1 - S_N$. So $\mathbb{P}(E_3) < S_N$.


\noindent{\sc Step 2. Bound $E_1$ and $E_2$. }
Given the event $E_4=\{\sup_t|\hat{V}^{(Y_{n+1})}_t - V_t|<C_\delta\}$, we have
\begin{align*}
    |\hat{G}(v) - \tilde{G}(v)| &\leq \frac{1}{N}\sum_{t=p+1}^{n+1}|\mathbbm{1}\{\hat{V}_t^{(Y_{n+1})} < v\} - \mathbbm{1}\{V_t < v\}|  \\
    & \leq \frac{1}{N} \sum_{t=p+1}^{n+1}\mathbbm{1}\{|V_t - v|\leq C_\delta\} + \frac{1}{N} \sum_t \mathbbm{1}\{|\hat{V}_t^{(Y_{n+1})} -V_t| \geq  C_\delta\}\\
    & \leq \sup_v\left|\frac{1}{N} \sum_{t=p+1}^{n+1}\mathbbm{1}\{|V_t - v|\leq C_\delta\} - \mathbb{P}(|V_{n+1} - v|\leq C_\delta)\right| + \mathbb{P}(|V_{n+1} - v|\leq C_\delta) \\
    & \leq  \sup_v\left|\frac{1}{N} \sum_{t=p+1}^{n+1} \left[ \mathbbm{1}\{V_t \leq v + C_\delta\} - \mathbbm{1}\{V_t \leq v - C_\delta\}\right] \right.\\
    &  \qquad - \mathbb{P}(V_{n+1}\leq v +  C_\delta) + \mathbb{P}(V_{n+1}\leq v -  C_\delta)\Bigg| + \mathbb{P}(|V_{n+1} - v|\leq C_\delta)  \\
    &  \leq  \sup_v\left|\tilde{G}(v+ C_\delta) -G(v+ C_\delta) - (\tilde{G}(v- C_\delta) -G(v- C_\delta))\right| \\
    & \qquad + \mathbb{P}(|V_{n+1} - v|\leq C_\delta)  \\
    & \leq 2\sup_{v}|\tilde{G}(v) -G(v)|+4C_\delta.
\end{align*}
Since $V_t$ is i.i.d, the Dvoretzky-Kiefer-Wolfowitz inequality implies that $$\mathbb{P}(E_2) \leq 2\exp(-8NC_\delta^2),$$ then we have 
\begin{align*}
\mathbb{P}(E_1) &= \mathbb{P}(E_1 \cap E_4) + \mathbb{P}(E_1 \cap E_4^c)\\
&\leq \mathbb{P}\left(\sup_v|\tilde{G}(v)-G(v)|\geq 2C_\delta\right) + \mathbb{P}(\sup_t|\hat{V}^{(Y_{n+1})}_t -V_t|\geq C_\delta)\\
&\leq 2 \exp(-8NC_\delta^2) + S_N.
\end{align*}
\noindent{\sc Step 3. Combine all bounds of $E_{k}$ for $k = 1, 2,3.$}
 We have
$$
\mathbb{P}(|\hat{G}(\hat{V}^{(Y_n+1)}_{n+1})-G(V_{n+1})|>12C_\delta) \leq 4\exp(-8NC_\delta^2) + 2S_N
$$
 with $C_\delta \asymp \frac{\log N}{\sqrt{Nh^{p}}}+h^2$. Plugging the upper bound back \eqref{eq: decom}, we have
 \begin{align*}
&|\mathbb{P}(Y_{n+1}\in \hat{C}^\text{MDCP}_{1-\alpha}(\bm{X}_n)) - (1-\alpha)| \leq 24C_\delta + 4\exp(-8NC_\delta^2) + 2S_N. 
\end{align*}

\subsection{Proof of Proposition \ref{Proposition:asymptoticMDCP}}
To show Proposition \ref{Proposition:asymptoticMDCP}, we consider the $\beta$-mixing coefficients defined in A.3.1 of \cite{francq2019garch}:
$$
\beta(\mathcal{F}_{-\infty}^j, \mathcal{F}_{\ell}^\infty):= \mathbb{E}  \sup_{B\in \mathcal{F}_{\ell}^\infty}|\mathbb{P}(B) -\mathbb{P}(B|\mathcal{F}_{-\infty}^j)|,$$ where $\mathcal{F}_{-\infty}^j$ and $\mathcal{F}_\ell^\infty$ are two $\sigma$-algebras. If we consider the kernel CDF estimator $\widehat{F}_S(Y|\bm{X})$ with only data $Y_{n-p-S},\ldots, Y_{1}$ where $S = o(n)$, the resulting lower bound and upper bound of the prediction interval $\widehat{L}^{\text{MDCP}}_{1-\alpha,S}$ and $\widehat{U}^{\text{MDCP}}_{1-\alpha,S}$ are measurable functions of $Y_{n-p-S},\ldots, Y_{1}$ and the candidate pair $(\bm{X}_n, Y_f)$. To show the asymptotic validity of the conditional coverage, we consider the lower bound $\widehat{L}^{\text{MDCP}}_{1-\alpha,S}$ of PI, the following analysis also holds for the upper bound. By the definition of geometric $\beta$-mixing coefficients, we can eventually show 
\begin{equation}\label{Eq:targeteq}
    \mathbb{E} |\mathbb{P}(\widehat{L}^{\text{MDCP}}_{1-\alpha,S}\leq c) -\mathbb{P}(\widehat{L}^{\text{MDCP}}_{1-\alpha,S}\leq c|\bm{X}_n)| \to 0,~\text{uniformly in}~c
\end{equation}
geometrically fast as $S\to\infty$. To show \eqref{Eq:targeteq}, we could consider showing a stronger result, i.e., $\mathbb{E} |\mathbb{P}(\widehat{L}^{\text{MDCP}}_{1-\alpha,S}\leq c) -\mathbb{P}(\widehat{L}^{\text{MDCP}}_{1-\alpha,S}\leq c|\mathcal{F}^{\infty}_{n})| \to 0$ uniformly in $c$. The crucial problem comes from the dependence between $\widehat{L}^{\text{MDCP}}_{1-\alpha,S}$ and $\mathcal{F}^{\infty}_{n}$ since the conformal prediction interval involves the data pair $(\bm{X}_n, Y_f)$. Thus, we consider any data pair $(\bm{x}_n,y_f)$ where $y_f$ is taken from $\mathcal{Y}_{\text{trial}}$ and observe that $\widehat{L}^{\text{MDCP}}_{1-\alpha,S}$ and $\mathcal{F}^{\infty}_{n}$ are in two separate $\sigma$-algebras now. In addition, we use the fact that $\beta(\mathcal{F}_{-\infty}^j, \mathcal{F}_{\ell}^\infty) = \beta(\mathcal{F}_{\ell}^\infty, \mathcal{F}_{-\infty}^j)$ which can be easily seen by considering the equivalent format of the $\beta$-mixing coefficients between two $\sigma$-algebras $(\mathcal{A}, \mathcal{B})$ defined on an appropriate probability space:
$$
\beta(\mathcal{A}, \mathcal{B}):= \frac{1}{2}\sup \sum_{i=1}^I \sum_{j=1}^J\left|P\left(A_i \cap B_j\right)-P\left(A_i\right) P\left(B_j\right)\right|;
$$
where the supremum is taken over all pairs of finite partitions $(A_1, \ldots, A_I)$ and $(B_1,  \ldots, B_J)$ such that $A_i \in \mathcal{A}$ for each $i$ and $B_j \in \mathcal{B}$ for each $j$. This fact implies that 
\begin{equation}\label{Eq:conditioanlx}
    \mathbb{E} \left[ \left|\mathbb{P}(\widehat{L}^{\text{MDCP}}_{1-\alpha,S}\leq c) -\mathbb{P}(\widehat{L}^{\text{MDCP}}_{1-\alpha,S}\leq c|\bm{x}_n, \mathcal{F}^{\infty}_{n+1})\right|\Bigg| \bm{X}_{n} = \bm{x}_n \right]\to 0,
\end{equation}
for any data pair $(\bm{x}_n,y_f)$ and $c$. Therefore, it still holds after taking one more expectation based on the tower property, i.e., 
\begin{equation}\label{Eq:unconditioanlx}
  \mathbb{E} \left[ \mathbb{E} \left[ \left|\mathbb{P}(\widehat{L}^{\text{MDCP}}_{1-\alpha,S}\leq c) -\mathbb{P}(\widehat{L}^{\text{MDCP}}_{1-\alpha,S}\leq c|\bm{x}_n, \mathcal{F}^{\infty}_{n+1})\right|\Bigg| \bm{X}_{n} = \bm{x}_n \right]\right]\to 0,
\end{equation}
for any $y_f\in\mathcal{Y}_{\text{trial}}$ and $c$, which is equivalently to what we want $$\mathbb{E} |\mathbb{P}(\widehat{L}^{\text{MDCP}}_{1-\alpha,S}\leq c) -\mathbb{P}(\widehat{L}^{\text{MDCP}}_{1-\alpha,S}\leq c|\bm{X}_n,\mathcal{F}^{\infty}_{n+1})| \to 0.$$

Finally, by Markov's inequality, we obtain that $\mathbb{P}(\widehat{L}^{\text{MDCP}}_{1-\alpha,S}\leq c|\bm{X}_n) \overset{p}{\to} \mathbb{P}(\widehat{L}^{\text{MDCP}}_{1-\alpha,S}\leq c)$, for any $y_f\in\mathcal{Y}_{\text{trial}}$ and $c$. To get the uniform convergence result, we notice that the domain of the whole series is bounded by $C_M$, as assumed in assumption A4. Therefore, the lower and upper bounds of a PI are in the compact of $[-C_M,C_M]$. Subsequently, the uniform convergence result can be concluded by applying the finite covering theorem. Moreover, since $S = o(n)$, $\widehat{L}^{\text{MDCP}}_{1-\alpha,S}$ and $\widehat{L}^{\text{MDCP}}_{1-\alpha}$ will be asymptotically equal with $n\to\infty$; where $\widehat{L}^{\text{MDCP}}_{1-\alpha}$ is the lower bound of the PI with smoothed kernel estimator based on whole original series. This convergence is in the a.s. sense. Thus, we can get $\mathbb{P}(\widehat{L}^{\text{MDCP}}_{1-\alpha}\leq c|\bm{X}_n) \to \mathbb{P}(\widehat{L}^{\text{MDCP}}_{1-\alpha,S}\leq c|\bm{X}_n)$ and $\mathbb{P}(\widehat{L}^{\text{MDCP}}_{1-\alpha,S}\leq c) \to \mathbb{P}(\widehat{L}^{\text{MDCP}}_{1-\alpha}\leq c)$ for any $c$ uniformly. Finally, we can conclude that $\mathbb{P}(\widehat{L}^{\text{MDCP}}_{1-\alpha}\leq c|\bm{X}_n) \overset{p}{\to} \mathbb{P}(\widehat{L}^{\text{MDCP}}_{1-\alpha}\leq c)$ uniformly in $c$. This convergence result is also true for the upper bound. Thus, we conclude that the coverage rate of the conditional prediction interval from MDCP is the same as the unconditional version asymptotically.

\subsection{Proof of Theorem \ref{thm: conditional}}
The idea to handle the effects of dependency within the data on the conditional prediction interval hinges on the $L^p$-$m$-approximator for the original series. We denote the conditional distribution estimator with $L^2$-$m$-approximator $\{Y_t^{(m)}\}$ as follows:
    \begin{equation}
\label{Eq:mapproxkernelEst}
 \hat{F}^{(m)}(y|\bm{x}) =\frac{\frac{1}{N}\sum_{i=p+1}^{n+1} W_h\left(\bm{X}_{i-1}^{(m)}, \bm{x}\right) K\left(\frac{y-Y^{(m)}_i}{h_0}\right)}{\overline{W}^{(m)}_h(\bm{x})}.
\end{equation}

To simplify the proof, we still take $h = h_0$. The bandwidth will be selected by cross-validation for the simulation and real-data studies. We consider the consistency between $\hat{F}^{(m)}(y|\bm{x})$ and $\hat{F}(y|\bm{x})$ for any $(\bm{x},y)$ in the joint domain. Let's denote:

$$\hat{F}(y|\bm{x}) -\hat{F}^{(m)}(y|\bm{x})=\frac{A_n-A_n^{(m)}}{\overline{W}_h(\bm{x})}+A_n^{(m)}\left(\frac{1}{\overline{W}_h(\bm{x})}-\frac{1}{\overline{W}^{(m)}_h(\bm{x})}\right);$$
where $$A_n := \frac{1}{N}\sum_{i=p+1}^{n+1} W_h\left(\bm{X}_{i-1}, \bm{x}\right) K\left(\frac{Y_i-y}{h}\right)$$ and $$A_n^{(m)} := \frac{1}{N}\sum_{i=p+1}^{n+1} W_h\left(\bm{X}_{i-1}^{(m)}, \bm{x}\right)K\left(\frac{Y^{(m)}_i-y}{h}\right).$$
By our assumption A5, $\overline{W}_h(\bm{X}_{t-1}^{(m)})$ are non-zero. We first analyze 
\begin{equation*}
    \begin{split}
        A_n - A_n^{(m)} &:= \frac{1}{N}\sum_{i=p+1}^{n+1}\left( W_h\left(\bm{X}_{i-1}, \bm{x}\right) K\left(\frac{Y_i-y}{h}\right) -  W_h\left(\bm{X}_{i-1}^{(m)}, \bm{x}\right)K\left(\frac{Y^{(m)}_i-y}{h}\right)\right)\\
        & = \frac{1}{N}\sum_{i=p+1}^{n+1} \Delta_i.
    \end{split}
\end{equation*}
In a similar analysis of \eqref{eq: V1}, we have $|\Delta_i|\leq \frac{C}{h^p}\sqrt{\frac{\left\|\bm{X}_{i-1} - \bm{X}_{i-1}^{(m)}\right\|^2 + \|Y_{i} - Y_{i}^{(m)}\|^2}{h^2}}$ with the smoothness assumption on $W(\cdot)$ and $K(\cdot)$. Therefore, 
$$|A_n - A_n^{(m)}|\leq \frac{1}{N}\sum_{i=p+1}^{n+1}\frac{C}{h^p}\sqrt{\frac{\left\|\bm{X}_{i-1} - \bm{X}_{i-1}^{(m)}\right\|^2+ \|Y_{i} - Y_{i}^{(m)}\|^2}{h^2}}.$$

Then, 

\begin{equation*}
\begin{split}
    \mathbb{E}( |A_n - A_n^{(m)}| ) &\leq \mathbb{E}\left( \frac{C}{h^p}\sqrt{\frac{\left\|\bm{X}_{i-1} - \bm{X}_{i-1}^{(m)}\right\|^2+ \|Y_{i} - Y_{i}^{(m)}\|^2}{h^2}} \right) \\
    & \leq \frac{C}{h^{p+1}}\mathbb{E}( \| \bm{X}_{i-1} - \bm{X}_{i-1}^{(m)} \|  + \|Y_{i} - Y_{i}^{(m)}\|) \\
    & \leq \frac{C}{h^{p+1}} (p+1)\delta(m);
\end{split}
\end{equation*}
where $\delta(m)$ is a sequence that convergences to $0$ as $m\to\infty$. By the assumption A6, we have $\delta(m) = o(h^{p+1})$, so that $\mathbb{E}( |A_n - A_n^{(m)}| ) \to 0$; the first inequality is due to the strict stationarity and the last two inequalities are due to the fact $\sqrt{a^2+b^2} \leq|a|+|b|$. Similarly, we can show $\mathbb{E}\left(\frac{1}{\overline{W}_h(\bm{x})}-\frac{1}{\overline{W}^{(m)}_h(\bm{x})}\right)\to 0$ with $\delta(m) = o(h^{p+1})$. 

Finally, by the Markov's inequality, we can conclude that $|\hat{F}(y|\bm{x}) -\hat{F}^{(m)}(y|\bm{x}) |\overset{p}{\to} 0$. To get the uniform consistency, we can repeat the finite covering technique used in proof of Lemma \ref{lem: mse of uhat}. Omitting the details, we have $\sup_{\bm{x},y}|\hat{F}(y|\bm{x}) -\hat{F}^{(m)}(y|\bm{x}) |\overset{p}{\to} 0$. Moreover, since $\widehat{V}_t^{\left(Y_{n+1}\right)}=\left|\widehat{U}_t^{\left(Y_{n+1}\right)}-1 / 2\right|$, where $\widehat{U}_t^{\left(Y_{n+1}\right)}=\widehat{F}^{\left(Y_{n+1}\right)}\left(Y_t \mid \boldsymbol{X}_{t-1}\right)$ for $t=p+1, \ldots, n+1$. By the continuous mapping theorem, we can conclude that $\sup_{\bm{x},y}|\widehat{V}_t^{\left(Y_{n+1}\right)} -\widehat{V}_t^{\left(m,Y_{n+1}\right)} |\overset{p}{\to} 0$; where $\widehat{V}_t^{\left(m,Y_{n+1}\right)} := \left|\widehat{U}_t^{\left(m,Y_{n+1}\right)}-1 / 2\right|$ and $\widehat{U}_t^{\left(m,Y_{n+1}\right)} := \widehat{F}^{\left(m, Y_{n+1}\right)}\left(Y_t \mid \boldsymbol{X}_{t-1}\right)$ which is the augmented CDF estimator with $\{Y_t^{(m)}\}$ series. 

Now consider two events, $\mathbbm{1}\left\{\widehat{V}_t^{(Y_{n+1})} \geq \widehat{V}_{n+1}^{(Y_{n+1})}\right\}$ and $\mathbbm{1}\left\{\widehat{V}_t^{(m, Y_{n+1})} \geq \widehat{V}_{n+1}^{(m, Y_{n+1})}\right\}$, it is not hard to show 
\begin{align*}
\mathbbm{1}\left\{\widehat{V}_t^{(m, Y_{n+1})} \geq \widehat{V}_{n+1}^{(m, Y_{n+1})}+o_p(1)\right\}& \leq \mathbbm{1}\left\{\widehat{V}_t^{(Y_{n+1})} \geq \widehat{V}_{n+1}^{(Y_{n+1})}\right\} \\
&\leq \mathbbm{1}\left\{\widehat{V}_t^{(m, Y_{n+1})} \geq \widehat{V}_{n+1}^{(m, Y_{n+1})}-o_p(1)\right\}
\end{align*}
for each $t$ asymptotically. Then, by considering the summation for $t = p+1,\ldots, n+1$ and taking $\mathbb{P}(\cdot|\bm{X}_{n})$, we have 
$$
\mathbb{P}\left(\frac{1}{N} \sum_{t=p+1}^{n+1} \mathbbm{1}\left\{\widehat{V}_t^{\left(Y_{n+1}\right)} \geq \widehat{V}_{n+1}^{\left(Y_{n+1}\right)}\right\}>\alpha \bigg|\boldsymbol{X}_n\right) \overset{p}{\to} \mathbb{P}\left(\frac{1}{N} \sum_{t=p+1}^{n+1} \mathbbm{1}\left\{\widehat{V}_t^{\left(m,Y_{n+1}\right)} \geq \widehat{V}_{n+1}^{\left(m, Y_{n+1}\right)}\right\}>\alpha \bigg| \boldsymbol{X}_n\right).$$

Subsequently, we apply the MDCP method with $\{Y_t^{(m)}\}$ series but ignore the latest $m+p+1$ values. More specifically, we make the transformation functions without the latest $m+p+1$ values and then determine the prediction interval in a conformal way. Since $m+p+1 = o(n)$, ignoring the latest $m+p+1$ values does not influence the asymptotic property of $\hat{F}^{(m)}(y|\bm{x})$. So we have 
\begin{align*}
&\mathbb{P}\left(\frac{1}{N} \sum_{t=p+1}^{n+1} \mathbbm{1}\left\{\widehat{V}_t^{\left(m,Y_{n+1}\right)} \geq \widehat{V}_{n+1}^{\left(m, Y_{n+1}\right)}\right\}>\alpha \bigg|\boldsymbol{X}_n\right) \\
\quad &\overset{p}{\to}  \mathbb{P}\left(\frac{1}{N} \sum_{t=p+1}^{n+1} \mathbbm{1}\left\{\widetilde{V}_t^{\left(m,Y_{n+1}\right)} \geq \widetilde{V}_{n+1}^{\left(m, Y_{n+1}\right)}\right\}>\alpha \bigg| \boldsymbol{X}_n\right);
\end{align*}
where $\widetilde{V}_t^{\left(m,Y_{n+1}\right)}$ is the variant of $\widehat{V}_t^{\left(m,Y_{n+1}\right)}$ after deleting latest $m+p+1$ values. 

Lastly, to show the asymptotic conditional validity, we note that 
$\widetilde{V}_t^{(m,Y_{n+1})}$ depends on $(\bm{X}_n, Y_f)$, 
and we handle this dependence as follows. For any fixed pair $(\bm{x}_n, y_f)$ with $y_f \in \mathcal{Y}_{\text{trial}}$, 
the conformity score $\widetilde{V}_{n+1}^{(m,y_f)}$ is a deterministic 
function of $(\bm{x}_n, y_f)$. Moreover, $\widetilde{V}_t^{(m,y_f)}$ 
for $t = p+1,\ldots, n-m-p$ is measurable with respect to 
$\mathcal{F}_{p+1}^{n-m-p}$, 
while $\bm{X}_n \in \mathcal{F}_{n-p+1}^{n}$. 
Since $\{Y_t^{(m)}\}$ is $m$-dependent and the time gap between 
$\mathcal{F}_{p+1}^{n-m-p}$ and $\mathcal{F}_{n-p+1}^{n}$ is $m+1 > m$, 
these two $\sigma$-algebras are strictly independent. By the fact that the $m$-dependence implies geometrically $\beta$-mixing, following the same argument as in the proof of Proposition
\ref{Proposition:asymptoticMDCP},  we have
\begin{equation*}
    \mathbb{P}\left(\frac{1}{N}\sum_{t=p+1}^{n+1} 
    \mathbbm{1}\left\{\widetilde{V}_t^{(m,Y_{n+1})} \geq 
    \widetilde{V}_{n+1}^{(m,Y_{n+1})}\right\} > \alpha 
    \;\middle|\; \bm{X}_n = \bm{x}_n\right) 
    \to (1-\alpha).
\end{equation*}
Since $m+p+1 = o(n)$, the truncated version is asymptotically 
equivalent to the full version, and hence
\begin{equation*}
    \mathbb{P}\left(\frac{1}{N}\sum_{t=p+1}^{n+1} 
    \mathbbm{1}\left\{\hat{V}_t^{(m,Y_{n+1})} \geq 
    \hat{V}_{n+1}^{(m,Y_{n+1})}\right\} > \alpha 
    \;\middle|\; \bm{X}_n = \bm{x}_n\right) 
    \to (1-\alpha)
\end{equation*}
for any fixed $\bm{x}_n$. We can then conclude that
\begin{equation*}
    \left|\mathbb{P}\left(\frac{1}{N}\sum_{t=p+1}^{n+1} 
    \mathbbm{1}\left\{\hat{V}_t^{(Y_{n+1})} \geq 
    \hat{V}_{n+1}^{(Y_{n+1})}\right\} > \alpha 
    \;\middle|\; \bm{X}_n\right) - (1-\alpha)\right| 
    \xrightarrow{p} 0,
\end{equation*}
which is exactly what we want to prove.


\subsection{Technical Lemmas}\label{Appendix:Lemma}

 \begin{lemma}[Dvoretzky-Kiefer-Wolfowitz inequality \citep{kosorok2008introduction}]
 \label{lem: DKW}
 Consider i.i.d sample $X_1, \ldots, X_n$ from a distribution with CDF $F(x)$, and let
$$
\widehat{F}_n(x)=\frac{1}{n} \sum_{i=1}^{n} \mathbbm{1}\{X_i \leq x\}
$$
denote the usual empirical distribution function. Then,
$$
\mathbb{P}\left(\sup _{x \in \mathbb{R}}\left\{\sqrt{n}\left|\widehat{F}_n(x)-F(x)\right|\right\}>\epsilon\right) \leq 2 e^{-2 \epsilon^2}
$$
for all $\epsilon>0$. 


    

\begin{definition}[Strong ($\alpha$-)mixing condition]
\label{def: alpha}
Suppose $\{X_i\}_{i \in \mathbb{Z}}$ is a sequence of random variables on a given probability space $(\Omega, \mathcal{F}, \mathbb{P})$. For $-\infty \leq j \leq \ell \leq \infty$, let $\mathcal{F}_j^{\ell}$ denote the $\sigma$-field of events generated by the random variables
$X _{ k }, j \leq k \leq \ell( k \in \mathbb{Z})$. For any two $\sigma$-fields $\mathcal{F}_{-\infty}^j$ and $\mathcal{F}_\ell^\infty$, define the ``measure of dependence"
$$
\alpha(\mathcal{F}_{-\infty}^j, \mathcal{F}_{\ell}^\infty):=\sup _{A \in \mathcal{F}_{-\infty}^j, B \in \mathcal{F}_{\ell}^\infty}|\mathbb{P}(A \cap B)-\mathbb{P}(A) \mathbb{P}(B)| .
$$
For the given random sequence $\{X_i\}$, for any positive integer $n$, define the dependence coefficient
$$
\alpha(n)=\alpha(X,n):=\sup _{j \in \mathbb{Z}} \alpha\left(\mathcal{F}_{-\infty}^j, \mathcal{F}_{j+n}^{\infty}\right).
$$
In the case where the given sequence $X$ is strictly stationary, the dependent coefficient has a simpler form
$$
\alpha(n)=\alpha(X, n):=\alpha\left(\mathcal{F}_{-\infty}^0, \mathcal{F}_n^{\infty}\right).
$$
\end{definition}

\begin{definition}[$\beta$-mixing condition]
  Suppose the same settings in Definition \ref{def: alpha}. For any two $\sigma$-fields $\mathcal{F}_{-\infty}^j$ and $\mathcal{F}_\ell^\infty$, define the $\beta$-mixing coefficients:
  $$
\beta(\mathcal{F}_{-\infty}^j, \mathcal{F}_{\ell}^\infty):= \mathbb{E}\sup_{B\in \mathcal{F}_{\ell}^\infty}|\mathbb{P}(B)-\mathbb{P}(B|\mathcal{F}_{-\infty}^j)|.
  $$
For the given random sequence $\{X_i\}$, we have similar definitions for $\beta(n)$.
\end{definition}

\begin{claim}[The relationship between $\alpha$- and $\beta$-mixing conditions]
For any two $\sigma$-fields $\mathcal{F}_{-\infty}^j$ and $\mathcal{F}_\ell^\infty$, these measures of dependence satisfy
$$
0\leq \alpha(\mathcal{F}_{-\infty}^j, \mathcal{F}_{\ell}^\infty)\leq 1/4, \quad 0 \leq \beta(\mathcal{F}_{-\infty}^j, \mathcal{F}_{\ell}^\infty) \leq 1
$$
and
$$
\alpha(\mathcal{F}_{-\infty}^j, \mathcal{F}_{\ell}^\infty) \leq 1/2\beta(\mathcal{F}_{-\infty}^j, \mathcal{F}_{\ell}^\infty).
$$
Thus, if a process $X_t$ is $\beta$-mixing, it implies $X_t$ is also $\alpha$-mixing.
\end{claim}


\begin{lemma}[Davydov's and Billingsley's inequality frome Lemma 1.1 and 1.2 of \cite{ibragimov1962some}]
\label{lem: davydov}
  Let $\{X_i\}_{i \in \mathbb{Z}}$ be a strictly stationary series which satisfies above $\beta$-mixing condition. If random variable $\xi$ is measurable w.r.t. $\mathcal{F}_{-\infty}^j$ and $\eta$ is measurable w.r.t. $\mathcal{F}_{j+n}^{\infty}$, and if $\mathbb{E}|\xi|^r < \infty$, $\mathbb{E}|\eta|^s < \infty$, with $r>1$, $s>1$, $\frac{1}{r} + \frac{1}{s} = 1$ then 
$$
|\mathrm{Cov}(\xi, \eta)| \leq2(\beta(n))^{1/r}\|X\|_r\|Y\|_s.
$$

If $\{X_i\}_{i \in \mathbb{Z}}$ is a strictly stationary series which satisfies above $\alpha$-mixing condition, and if $|\xi| < C_1$, $|\eta| < C_2$ for some constants $C_1$ and $C_2$, then 
$$
|\mathrm{Cov}(\xi, \eta)| \leq4\alpha(n)C_1C_2.
$$

\end{lemma}

\begin{lemma}[Theorem 1.3 of \cite{bosq2012nonparametric}]
\label{lem: upper bound}
    Let $\{X_t\}_{t \in \mathbb{Z}}$ be a zero-mean real-valued strictly stationary bounded process satisfying $\alpha$-mixing condition. Then for each integer $q \in \left[1, \frac{n}{2}\right]$ and each $\varepsilon > 0$, we have
$$
\mathbb{P}\!\left( \left| \sum_{t=1}^n X_t \right| > n \varepsilon \right)
\leq
4 \exp\!\left(-\frac{\varepsilon^2}{8 v^2(q)}\, q\right)+ 22 \left(1 + \frac{4 \|X_0\|_\infty}{\varepsilon}\right)^{1/2} q\, \alpha\!\left( \left\lfloor \frac{n}{2q} \right\rfloor \right),
$$
where
$$
v^2(q) = \frac{2}{(\frac{n}{2q})^2}
\mathrm{Var}\bigg(\sum_{t=1}^{\left\lfloor \frac{n}{2q} \right\rfloor + 1} X_t\bigg)+\frac{\varepsilon \|X_0\|_\infty}{2},
$$
and $\alpha\!\left( \left\lfloor \frac{n}{2q} \right\rfloor \right)$ is the strong mixing coefficient
of order $\left\lfloor \frac{n}{2q} \right\rfloor$. For simplicity, we denote $\mathrm{Var}\bigg(\sum_{t=1}^{\left\lfloor \frac{n}{2q} \right\rfloor + 1} X_t\bigg)$ by $\sigma^2(q)$.
\end{lemma}
\subsubsection{Proof of Lemma \ref{lem: mse of uhat}}

\begin{proof}[Proof of Lemma \ref{lem: mse of uhat}]
    The proof is motivated by Proof of Theorem 1.4 in \cite{li2007nonparametric}. Throughout this proof, $C_1,\ldots  C_{12}$ represent different constants. We consider the decomposition
    \begin{equation}
        \begin{aligned}
 \sup _{\bm{x}, y} {|\hat{F}(y|\bm{x})-F(y|\bm{x})|} &=\sup _{\bm{x}, y}\left|\frac{1}{\overline{W}_h(\bm{x})}(\hat{F}(\bm{x},y)-\overline{W}_h(\bm{x}) F(y|\bm{x}))\right| \\
& =\sup _{\bm{x}, y}\left|\frac{1}{\overline{W}_h(\bm{x})}(\hat{F}(\bm{x}, y)-\mathbb{E}\hat{F}(\bm{x}, y)+\mathbb{E}\hat{F}(\bm{x}, y)-\mathbb{E}(\overline{W}_h(\bm{x}) F(y|\bm{x}))\right.\\
&\qquad +\mathbb{E}(\overline{W}_h(\bm{x}) F(y|\bm{x}))-\overline{W}_h(\bm{x}) F(y|\bm{x}))\Bigg| \\
& \leq C_1 \cdot \sup _{\bm{x}, y}\left|\hat{F}(\bm{x}, y)-\mathbb{E}\hat{F}(\bm{x}, y)+\mathbb{E}\hat{F}(\bm{x}, y)-\mathbb{E}(\overline{W}_h(\bm{x}) F(y|\bm{x}))\right.\\
&\qquad +\mathbb{E}(\overline{W}_h(\bm{x}) F(y|\bm{x}))-\overline{W}_h(\bm{x}) F(y|\bm{x})\Bigg| \\
& \leq C_1 \cdot[\sup _{\bm{x}, y}|\hat{F}(\bm{x}, y)-\mathbb{E}\hat{F}(\bm{x}, y)| +\sup _{\bm{x}, y}|\mathbb{E}\hat{F}(\bm{x}, y)-\mathbb{E}(\overline{W}_h(\bm{x}) F(y|\bm{x}))|\\
&\qquad +\sup _{\bm{x}, y}|\mathbb{E}(\overline{W}_h(\bm{x}) F(y|\bm{x}))-\overline{W}_h(\bm{x}) F(y|\bm{x})|] \\
& =C_1 \cdot\left(Q_1+Q_2+Q_3\right) ;
\end{aligned}
    \end{equation}
    where 
    $$
    \begin{aligned}
        Q_1 &:=\sup _{\bm{x}, y}|\hat{F}(\bm{x}, y)-\mathbb{E}\hat{F}(\bm{x}, y)|,\\
    Q_2 &:=\sup _{\bm{x}, y}|\mathbb{E}\hat{F}(\bm{x}, y)-\mathbb{E}(\overline{W}_h(\bm{x}) F(y|\bm{x}))|,\\
    Q_3&:=\sup _{\bm{x}, y}|\mathbb{E}(\overline{W}_h(\bm{x})  F(y|\bm{x})) - \overline{W}_h(\bm{x}) F(y|\bm{x})|.
    \end{aligned}
    $$
    Then, we will analyze the order of $Q_1, Q_2$ and $Q_3$ one by one.
    
\noindent{\sc Step 1. Consider the bound of $Q_1$.}
Under Assumption A4, the domain of $X$ and $Y$ is a compact set $\left[-C_M, C_M\right]^{p+1}$. Define the space $\mathcal{S}:= \left[-C_M, C_M\right]^{p+1}$ which can be covered by a finite number $J_n$ cubes $\Xi_k$ with centers $({x}_{k, n},y_{k,n})$ for $k= 1, \ldots, J_n$ ; $J_n$ is a constant which depends on the dimension $p+1$, the constant $C_M$, and the length of the cube $l_n$. Then, we have
 \begin{equation}
    \begin{aligned}
\sup _{\bm{x} \times y \in \mathcal{S}}|\hat{F}(\bm{x}, y)-\mathbb{E}\hat{F}(\bm{x}, y)|= & \max _{1 \leq k \leq J_n} \sup _{\bm{x} \times y \in \mathcal{S} \cap \Xi_k}|\hat{F}(\bm{x}, y)-\mathbb{E}\hat{F}(\bm{x}, y)| \\
\leq & \max _{1 \leq k \leq J_n} \sup _{\bm{x} \times y \in \mathcal{S} \cap \Xi_k}\left|\hat{F}(\bm{x}, y)-\hat{F}\left(\bm{x}_{k, n},y_{k,n}\right)\right| \\
& +\max _{1 \leq k \leq J_n}\left|\hat{F}\left(\bm{x}_{k,n},y_{k,n}\right)-\mathbb{E}\hat{F}\left(\bm{x}_{k,n},y_{k,n}\right)\right| \\
& +\max _{1 \leq k \leq J_n} \sup _{\bm{x} \times y \in \mathcal{S} \cap \Xi_k}\left|\mathbb{E} \hat{F}\left(\bm{x}_{k,n},y_{k,n}\right) -\mathbb{E}\hat{F}(\bm{x}, y)\right| \\
:= & V_1+V_2+V_3 .
\end{aligned}
\end{equation}
\noindent{\sc Step 1.1 Derive the term $V_1$.}
Denote $G(\bm{X}_{t-1}, Y_t,\bm{x},y) =  W_h(\bm{X}_{t-1},\bm{x})K\left(\frac{y-Y_t}{h_0}\right)$ and $N = n-p+1$. Thus, we have the upper bound of $V_1$ by
\begin{equation}
\label{eq: V1}
    \begin{aligned}
    & \max _{1 \leq k \leq J_n} \sup _{\bm{x} \times y \in \mathcal{S} \cap \Xi_k}\left|\hat{F}(\bm{x}, y)-\hat{F}\left(\bm{x}_{k, n},y_{k,n}\right)\right|  \\
    &=\max _{1 \leq k \leq J_n} \sup _{\bm{x} \times y \in \mathcal{S} \cap \Xi_k}\frac{1}{N}\left|\sum_{t=p+1}^{n+1} G(\bm{X}_{t-1},Y_t,\bm{x},y)- \sum_{t=p+1}^{n+1} G(\bm{X}_{t-1},Y_t,\bm{x}_{k,n},y_{k,n})\right| \\
    & \leq \frac{C_2}{Nh^p} \max_{1 \leq k\leq J_n }\sup_{\bm{x}\times y\in \mathcal{S}\cap \Xi_k} \sum_{t=p+1}^{n+1} \sqrt{L_w^2+L_K^2}\sqrt{\frac{\|\bm{x}-\bm{x}_{k,n}\|^2}{h^2} +\frac{(y-y_{k,n})^2}{h_0^2}}\\
    & \leq C_3\frac{\sqrt{h_0^2+h^2}}{h^{p+1}h_0}l_n;
    \end{aligned}
\end{equation}
where the first equality is due to the definition of kernel estimator; in the second equality, $w$ in $W_h$ is $L_w$-Lipschitz continuous, $K$ is $L_K$-Lipschitz continuous on the closed domains from Assumption A4, $w$ and $K$ are bounded on their whole domain from Assumption A5 and the basic inequality
$$
|f(x_1)g(y_1)-f(x_2)g(y_2) |\leq \sqrt{M_f^2L_g^2+M_g^2L_f^2}\cdot \sqrt{(x_1-x_2)^2+(y_1-y_2)^2},
$$
if $f$ is $L_f$-Lipschitz continuous and bounded by $M_f$ and $g$ is $L_g$-Lipschitz continuous and bounded by $M_g$; and the last inequality is due to the finite cover of the space $\mathcal{S}$; $C_3$ is constant independent of $l_n$, $h$, $h_0$ and $p$.

\noindent{\sc Step 1.2 Derive the upper bound of $V_3$.}

For the term $V_3:=\max _{1 \leq k \leq J_n} \sup _{\bm{x} \times y \in \mathcal{S} \cap \Xi_k}\left|\mathbb{E}\left(\hat{F}\left(\bm{x}_{k, n},y_{k,n}\right)\right)-\mathbb{E}(\hat{F}(\bm{x}, y))\right|$, we have:
\begin{equation}
\label{eq: V3}
\begin{aligned}
& \max _{1 \leq k \leq J_n} \sup _{\bm{x} \times y \in \mathcal{S} \cap \Xi_k}\left|\mathbb{E}\left(\hat{F}\left(\bm{x}_{k, n},y_{k,n}\right)\right)-\mathbb{E}(\hat{F}(\bm{x}, y))\right| \\
&\leq \max _{1 \leq k \leq J_n} \sup _{\bm{x} \times y \in \mathcal{S} \cap \Xi_k} \mathbb{E}\left[\left|\hat{F}\left(\bm{x}_{k, n},y_{k,n}\right)-\hat{F}(\bm{x}, y)\right|\right] \\
& \leq \max _{1 \leq k \leq J_n} \mathbb{E}\left[\sup _{\bm{x} \times y \in \mathcal{S} \cap \Xi_k}\left|\hat{F}\left(\bm{x}_{k, n},y_{k,n}\right)-\hat{F}(\bm{x}, y)\right|\right] \\
& \leq C_3\frac{\sqrt{h_0^2+h^2}}{h^{p+1}h_0}l_n;
\end{aligned}
\end{equation}
where the last inequality is due to the inequality \eqref{eq: V1}. 

\noindent{\sc Step 1.3 Derive the upper bound of $V_2$.}
For $$V_2:=\max _{1 \leq k \leq J_n}\left|\hat{F}\left(\bm{x}_{k, n},y_{k,n}\right)-\mathbb{E}\left(\hat{F}\left(\bm{x}_{k, n},y_{k,n}\right)\right)\right|,$$ we can define a function 
$$
\begin{aligned}
\Psi_n(\bm{x},y)&:= \hat{F}(\bm{x}, y)-\mathbb{E}\hat{F}(\bm{x},y)\\
&= \frac{1}{N}\left[\sum_{t=p+1}^{n+1}W_h(\bm{X}_{t-1},\bm{x})K\left(\frac{y-Y_t}{h_0}\right) -\mathbb{E}W_h(\bm{X}_{t-1},\bm{x})K\left(\frac{y-Y_t}{h_0}\right)\right]:=\frac{1}{N}\sum_{t=p+1}^{n+1}\psi_{t,n}.
\end{aligned}
$$
Then, for any $\eta>0$, we can write
\begin{equation}
\begin{aligned} 
\mathbb{P}\left(V_2>\eta\right) & =\mathbb{P}\left(\max _{1 \leq k \leq J_n}\left|\hat{F}\left(\bm{x}_{k,n}, y_{k,n}\right)-\mathbb{E}\left(\hat{F}\left(\bm{x}_{k,n}, y_{k,n}\right)\right)\right|>\eta\right) \\ 
& =\mathbb{P}\left(\max _{1 \leq k \leq J_n}\left|\Psi_n\left(\bm{x}_{k,n}, y_{k,n}\right)\right|>\eta\right) \\ 
& \leq \mathbb{P}\left(\left|\Psi_n\left(\bm{x}_{1,n}, y_{1,n}\right)\right|>\eta \bigcup\left|\Psi_n\left(\bm{x}_{2,n}, y_{2,n}\right)\right|>\eta \bigcup \cdots\left|\Psi_n\left(\bm{x}_{J_n,n}, y_{J_n,n}\right)\right|>\eta\right) \\ 
& \leq J_n \max _{1 \leq k \leq J_n } \mathbb{P}\left(\left|\Psi_n(\bm{x}_{k, n},y_{k,n})\right|>\eta\right) .
\end{aligned}
\end{equation}

The following proof works for any $(\bm{x}_{k,n},y_{k,n})$. Denote $C_4: =\sup_{\bm{x}}|\prod_{s=1}^p w(\bm{x}_s)|.$ Denote that $|G(\bm{X}_{t-1}, Y_t,\bm{x},y) |\leq b := C_4h^{-p}$. So, $\psi_{n,t} \leq C_5h^{-p}$. Moreover, by our assumption, there exist $\gamma>0$ and $\rho\in(0,1)$ such that $\alpha(k)\leq 1/2\beta(k) \leq {\gamma}/{2}\rho^k$ for $k\geq 1.$ 
By the definition of $\psi_{n,t}$, we have
$$
\mathrm{Var}(\psi_{n,t}) =\mathrm{Var}(G(\bm{X}_{t-1}, Y_t,\bm{x},y)) = \mathbb{E}(G^2(\bm{X}_{t-1}, Y_t,\bm{x},y))- \mathbb{E}^2(G(\bm{X}_{t-1}, Y_t,\bm{x},y)).
$$
For simplification, denote $G(\bm{X}_{t-1}, Y_t,\bm{x},y)$ by $G$. By the definition of $G$, we consider the second moment $\mathbb{E}[G^2]$
\begin{align*}
    \mathbb{E}[G^2] &= \frac{1}{h^{2p}}\iint \prod_{s=1}^p w^2\left(\frac{u_s-x_{s}}{h}\right) K^2\left(\frac{y-v}{h_0}\right) f_{\bm{X},Y}(\bm{u},v)d\bm{u} dv\\
   & =\frac{1}{h^{p}}\iint \prod_{s=1}^pw^2(u'_s)K^2\left(\frac{y-v}{h_0}\right) f_{\bm{X},Y}(\bm{x}h+u',v)du'dv = O\left(\frac{1}{h^{p}}\right)
\end{align*}
and first moment $\mathbb{E}[G]$ 
\begin{align*}
    \mathbb{E}[G] &= \frac{1}{h^{p}}\iint \prod_{s=1}^p w\left(\frac{u_s-x_{s}}{h}\right) K\left(\frac{y-v}{h_0}\right) f_{\bm{X},Y}(\bm{u},v)d\bm{u} dv\\
   & =\iint \prod_{s=1}^pw(u'_s)K\left(\frac{y-v}{h_0}\right) f_{X,Y}(\bm{x}h+\bm{u}',v)d\bm{u}'dv = O\left(1\right).
\end{align*}
Using Davydov’s inequality mentioned in Lemma \ref{lem: davydov} with $r = 1/2$ and $s = 1/2$, we have
\begin{align*}
|\mathrm{Cov}(\psi_{i,n}, \psi_{j,n}) |&\leq 4\beta(|i-j|)^{1/2}\|\psi_{i,n}\|_2\|\psi_{j,n}\|_2 \leq C_6\rho^{-|i-j|/2}\gamma/h^{p}.
\end{align*}
Therefore, by taking $q = \sqrt{N}h^{-\frac{p}{2}}$, for $n$ large enough, s.t., $q\in \left[1, \frac{N}{2}\right]$, we have 
\begin{align*}
    v^2(q) = \frac{8q^2}{N^2}\sigma^2(q) + \frac{A_6\eta}{2h^p}.
\end{align*}
Without loss of generality, we assume $k:= \left\lfloor \frac{N}{2q} \right\rfloor$ is an integer, then we have 
\begin{equation*}
    \begin{split}
        \sigma^2(q) = 
&\sum_{t=1}^{k+1} \operatorname{Var}\left(\psi_{t,n}\right)+2 \sum_{1 \leq i<j \leq k+1} \operatorname{Cov}\left(\psi_{i,n}, \psi_{j,n}\right) \\
&\leq O\left(\frac{k}{h^p}\right) + 2
\sum_{i=1}^{k} \sum_{l=1}^{k+1-i}\left|\operatorname{Cov}\left(\psi_{i,n}, \psi_{l+i, n}\right)\right|\\
&\leq O\left(\frac{k}{h^p}\right)  + O(h^{-p} \sum_{i=1}^{k}(k+1-i) \rho^{-i/2}) \\
&  \leq O\left(\frac{k}{h^p}\right)  +
 O( k h^{-p} \sum_{i=1}^{\infty} \rho^{-i/2}) =  O( k h^{-p}).
    \end{split}
\end{equation*}
With the $\eta$ chosen as below, the above inequality implies that
$
v^2(q) = O\left(\frac{A_6\eta}{h^{p}}\right).
$
Hence by Lemma \ref{lem: upper bound}, we have,
\begin{align*}
    \mathbb{P}\left(\frac{1}{N}\left|\sum_{t=p+1}^{n+1} \psi_{t,n} \right|>\eta\right) \leq 4\exp\left(-\frac{\eta^2}{8v^2(q)}q\right)+22\left(1+\frac{4A_5h^{-p}}{\eta}\right)^{1/2}\frac{q\gamma}{2} \rho^{[\frac{N}{2q}]},
\end{align*}
which implies 
$$
 \mathbb{P}\left(\frac{1}{N}\left|\sum_{t=p+1}^{n+1} \psi_{n,t} \right|>\eta\right) \leq 4\exp(-C_7(Nh^p)^{1/2}\eta)+C_8N^{1/2}h^{-p} \eta^{-1/2} \rho^{[\frac{N}{2q}]} ;
$$
By setting $\eta =\frac{\log N}{\sqrt{Nh^p}} $, we obtain the bound
$$
4N^{-C_7} + C_8N^{3/4}h^{-3p/4} (\log N)^{-1/2} \rho^{\sqrt{Nh^p}} \lesssim N^{-C_9},
$$
which tends to zero for some $C_9$ when $Nh^p$ is sufficiently large. The above inequality can be seen by writing 
\begin{equation*}
    \begin{split}
        C_8N^{3/4}h^{-3p/4} (\log N)^{-1/2} \rho^{\sqrt{Nh^p}} &= C e^{\frac{3}{4}\log N}\cdot e^{-\frac{3}{4}\log h^p}\cdot (\log N)^{-1/2} \cdot e^{\sqrt{Nh^p}\log \rho} \\
        &= C_8 e^{\frac{3}{4}\log \frac{N}{h^p}}\cdot (\log N)^{-1/2} \cdot e^{-C_{10}\sqrt{Nh^p}} \\
        & = C_8 (\log N)^{-1/2} \cdot e^{\frac{3}{4}\log \frac{N}{h^p} - C_{10}\sqrt{Nh^p}}.
    \end{split}
\end{equation*}
Let $h^p = N^{-\tau}$, where $0<\tau<1$ so the assumption $Nh^p \to \infty$ satisfies and $N^{1-\tau-\frac{2\tau}{p}-C_9}\to 0$. Then, we have $CN^{3/4}h^{-3p/4} (\log N)^{-1/2} \rho^{\sqrt{Nh^p}}  = C_8 (\log N)^{-1/2} \cdot e^{\frac{3}{4}\log N^{1+\tau} - C_{10}\sqrt{N^{1-\tau}}}$, which converges to 0 as exponentially fast rate. All in all, we have $\mathbb{P}(V_2 >\eta) \leq \frac{2J_n}{N^{C_9}}.$

\noindent{\sc Step 1.4. Combine all terms $V_1, V_2$ and $V_3$. }
To incorporate all terms $V_1, V_2$ and $V_3$ together, we can take $h_0 \asymp h$ for simplicity and $l_n=\log (N)\left(h^{p+2}/N\right)^{1 / 2}$. Then, inequalities \eqref{eq: V1} and \eqref{eq: V3} satisfy
\begin{equation}
\label{eq: V1 V3}
\begin{aligned}
& V_1:=\max _{1 \leq k \leq J_n} \sup _{\bm{x} \times y \in \mathcal{S} \cap \Xi_k}\left|\hat{F}(\bm{x}, y)-\hat{F}\left(\bm{x}_{k, n}, y_{k,n}\right)\right| \leq \frac{C_{11}}{h^{p+1}} l_n=C_{11}\log N\left(\frac{1}{N h^{p}}\right)^{1 / 2}; \\
& V_3:=\max _{1 \leq k \leq J_n} \sup _{\bm{x} \times y \in \mathcal{S} \cap \Xi_k}\left|\mathbb{E}\left(\hat{F}\left(\bm{x}_{k, n}, y_{k,n}\right)\right)-\mathbb{E}(\hat{F}(\bm{x}, y))\right| \leq \frac{C_{11}}{h^{p+1}} l_n=C_{11}\log N\left(\frac{1}{N h^{p}}\right)^{1 / 2} .
\end{aligned}
\end{equation}
Without loss of generality, we assume that $C_{12}$ is a large value such that $C_{12}>\max \left(C_{11}, 1\right)$. 
Combining the inequalities in $V_2$ and \eqref{eq: V1 V3}, we can obtain the upper bound for $Q_1$ with probability at least $1-2J_n N^{-C_9}$,
\begin{equation}
    \label{eq: Q1}
    \begin{aligned} Q_1 & :=\sup _{\bm{x} \times y \in \mathcal{S}}|\hat{F}(\bm{x}, y)-\mathbb{E}(\hat{F}(\bm{x}, y))| \leq V_1+V_2+V_3  \leq 3 C_{12}\log N\left(\frac{1}{N h^{p}}\right)^{1 / 2};\end{aligned}
\end{equation}
where $J_n \asymp \frac{1}{l_n^{p+1}}$.

\noindent{\sc Step 2. Consider the bound of $Q_3$.} For $Q_3:=\sup _{\bm{x}, y}|\mathbb{E}(\overline{W}_h(\bm{x}) F(y|\bm{x}))-\overline{W}_h(\bm{x}) F(y|\bm{x})|$, we have
$$
\sup _{\bm{x}, y}|\mathbb{E}(\overline{W}_h(\bm{x}) F(y|\bm{x}))-\overline{W}_h(\bm{x}) F(y|\bm{x})| \leq \sup _{\bm{x}}|\overline{W}_h(\bm{x})-\mathbb{E}(\overline{W}_h(\bm{x}))| ;
$$
then, with probability at least $1- 2J_n'N^{-C_9^{'}}$, we have
\begin{equation}
    \label{eq: Q3}
    \begin{aligned} 
    Q_3:=\sup _{\bm{x}, y}|\mathbb{E}(\overline{W}_h(\bm{x}) F(y|\bm{x}))-\overline{W}_h(\bm{x}) F(y|\bm{x})| \leq 3C_{12}^{'} \frac{\log N}{\sqrt{Nh^p}};
    \end{aligned}
\end{equation}
where $C_9^{'}$ and $C_{12}^{'}$ are constants similar to $C_9$ and $C_{12}$ above, respectively, by taking $l'_n = O\left(\log (N)\left(h^{p}/N\right)^{1 / 2}\right).$

\noindent{\sc Step 3. Consider the bound of $Q_2$.}
At first, a third-order differentiable $f(\bm{x})$ has a second-order $p$-dimensional Taylor expansion at $x$
$$
f(\bm{x}+ \Delta) = f(\bm{x}) + \nabla f(\bm{x} )^{\top} \Delta +\frac{1}{2}\Delta^{\top} \nabla^2 f(\bm{x} + t \Delta ) \Delta;
$$
where $t \in (0,1)^p$. Now we will use the definition of expectation for multivariate functions to discuss $\mathbb{E}\hat{F}(\bm{x},y)$.

For $$
\mathbb{E}\hat{F}(\bm{x}, y)= \frac{1}{Nh^p}\sum_{t=p+1}^{n+1}\mathbb{E}\left[\prod_{s=1}^p w\left(\frac{X_{t-1,s}-x_s}{h}\right)K\left(\frac{y-Y_t}{h_0}\right)\right],$$ we have
\begin{align*}
         \mathbb{E}\hat{F}(\bm{x},y) &
         = \frac{1}{h^p}\iint\prod_{s=1}^p w\left(\frac{u_s-x_s}{h}\right)K\left(\frac{y-u_y}{h_0}\right)f_{\bm{X},Y}(\bm{u},u_y)d\bm{u}du_y\\
         & = \iint \prod_{s=1}^p w(v_s)K\left(\frac{y-u_y}{h_0}\right)f_{\bm{X},Y}(\bm{x}+h\bm{v}, u_y) d\bm{v}du_y\\
         & \stackrel{(a)}{=}\iint \prod_{s=1}^pw(v_s) K\left(\frac{y-u_y}{h_0}\right) \cdot \bigg(f_{\bm{X},Y}(\bm{x},u_y) + h \bm{v}^\top \nabla_{\bm{x}} f_{\bm{X},Y}(\bm{x},u_y) \\
       &  \qquad \qquad + \frac{1}{2}h^2\bm{v}^\top \nabla^2_{\bm{x}} f_{\bm{X},Y}(\bm{\xi}, u_y) \bm{v}d\bm{v} du_y\bigg)\\
       & = \iint \prod_{s=1}^p w(v_s) K\left(\frac{y-u_y}{h_0}\right)f_{\bm{X},Y}(\bm{x},u_y) d\bm{v}du_y\\
       &\qquad + \iint \prod_{s=1}^pw(v_s) K\left(\frac{y-u_y}{h_0}\right)h \bm{v}^\top \nabla_{\bm{x}} f_{\bm{X},Y}(\bm{x},u_y) d\bm{v}du_y\\
       &\qquad +\frac{1}{2}\iint \prod_{s=1}^pw(v_s) K\left(\frac{y-u_y}{h_0}\right)h^2 \bm{v}^\top \nabla^2_{\bm{x}} f_{\bm{X},Y}(\bm{\xi}, u_y)\bm{v}d\bm{v} du_y \\
       & \stackrel{(b)}{=} \prod_{s=1}^p\int w(v_s)dv_s \int K\left(\frac{y-u_y}{h_0}\right)f_{\bm{X},Y}(\bm{x},u_y) du_y \\
       &\qquad +\frac{1}{2}h^2 \iint \prod_{s=1}^pw(v_s) K\left(\frac{y-u_y}{h_0}\right)\bm{v}^\top \nabla^2_{\bm{x}}f_{\bm{X},Y}(\bm{\xi}, u_y)\bm{v}d\bm{v} du_y \\
       & = f_X(\bm{x}) \int K\left(\frac{y-u_y}{h_0}\right)f_{Y|X}(u_y|\bm{x}) du_y\\
       & \qquad +\frac{1}{2}h^2 \iint \prod_{s=1}^pw(v_s) K\left(\frac{y-u_y}{h_0}\right)\bm{v}^\top \nabla^2_{\bm{x}}f_{\bm{X},Y}(\bm{\xi}, u_y)\bm{v}d\bm{v} du_y\\
       & \stackrel{(c)}{=} f_X(\bm{x}) F(y|\bm{x}) +\frac{1}{2}h^2 \iint \prod_{s=1}^pw(v_s) \bm{v}^\top \nabla^2_{\bm{x}}f_{\bm{X},Y}(\bm{\xi}, u_y)\bm{v}d\bm{v} du_y +O(h^2);
\end{align*}
where in (a), for each $s$, $\bm{\xi}_s$ in $\bm{\xi} =({\xi}_1, \ldots, {\xi}_p)$ ${\xi}_s$ is between ${x}_s$ and ${x}_s+h{v}_s$; (b) is derived from $\int vw(v)dv=0$ in Assumption A5; and (c) is derived from
\begin{align*}
   & \int K\left(\frac{y-u_y}{h_0}\right)f_{Y|X}(u_y|\bm{x}) du_y \\
    &= K\left(\frac{y-u_y}{h_0}\right)F(u_y|\bm{x})|_{-\infty}^{+\infty} - \int F(u_y|\bm{x})\left(-\frac{1}{h_0}k\left(\frac{y-u_y}{h_0}\right)\right)du_y\\
    & =\frac{1}{h_0}\int F(u_y|\bm{x})k\left(\frac{y-u_y}{h_0}\right)du_y = \int F(y-h_0 z|\bm{x})k(z)dz \\
    & = \int F(y|\bm{x})k(z)dz- \int h_0 zf(y|\bm{x}) k(z)dz +\frac{1}{2}h_0^2 \int z^2f'(\xi_y|\bm{x}) k(z) dz \\
    &= F(y|\bm{x}) + O(h^2);
\end{align*}
where $\int zk(z)dz =0$, $\int z^2 k(z)dz<\infty$, and $\sup_{\bm{x},y}|f'(y|\bm{x})|<\infty$.

And for $\mathbb{E}\overline{W}_h(\bm{x})F(y|\bm{x})$, we have
\begin{equation}
    \begin{aligned}
          \mathbb{E}(\overline{W}_h(\bm{x})F(y|\bm{x}))&= F(y|\bm{x}) \cdot \frac{1}{h^p}\int \prod_{s=1}^pw\left(\frac{u_s-x_s}{h}\right)f_{\bm{x}}(u_s)du\\
          & = F(y|\bm{x}) \int \prod_{s=1}^p w(v_s) f_{\bm{x}}({\bm{x}}+h\bm{v})d\bm{v}\\
          & = F(y|\bm{x}) \int \prod_{s=1}^p w(v_s)\bigg(f_{\bm{x}}(\bm{x})-h\bm{v}^\top \nabla_{\bm{x}} f_{\bm{x}}(\bm{x})+\frac{1}{2}h^2 \bm{v}^\top \nabla^2_{\bm{x}} f_{\bm{x}}(\bm{\xi}')\bm{v}\bigg)d\bm{v} \\
          & = F(y|\bm{x})f_X(\bm{x})+\frac{h^2}{2}F(y|\bm{x})\int \prod_{s=1}^p w(v_s) \bm{v}^\top \nabla^2_{\bm{x}} f_X(\bm{\xi}')\bm{v}d\bm{v}.
    \end{aligned}
\end{equation}

Thus, 
\begin{align*}
    Q_2&:=\sup_{\bm{x},y}|\mathbb{E}(\hat{F}(\bm{x},y) - \mathbb{E}(\overline{W}_h(\bm{x})F(y|\bm{x}))|\\
    & = \frac{h^2}{2}\sup_{\bm{x},y}\bigg| \iint \prod_{s=1}^pw(v_s) \bm{v}^\top \nabla^2_{\bm{x}}f_{\bm{X},Y}(\bm{\xi}, u_y)\bm{v}d\bm{v} du_y \\
    & \qquad - F(y|\bm{x})\int \prod_{s=1}^p w(v_s) \bm{v}^\top \nabla^2_{\bm{x}} f_X(\bm{\xi}')\bm{v}d\bm{v}\bigg|+O(h^2)\\
& \leq \frac{h^2}{2} \sup_{\bm{x},y} |M_{\bm{x},y} + F(y|\bm{x})\cdot M_{x}| \iint \prod_{s=1}^pw(v_s) \bm{v}^\top \bm{v}d\bm{v} \leq O( h^2);
\end{align*}
where $M_{\bm{x},y} := \sup_{\bm{x},y} |\nabla^2_{\bm{x}}f_{\bm{X},Y}(\bm{x}, y)|$, $M_{x}:=\sup_x |\nabla^2_{\bm{x}}f_X(\bm{x})|$.

\noindent{\sc Step 4. Combine the terms $Q_1, Q_2, Q_3$.} With Fréchet inequalities, we have
$$
\sup_{\bm{x},y}|\hat{F}(y|\bm{x}) - F(y|\bm{x})| \leq C_1(Q_1+Q_2+Q_3) \leq O\left( \frac{\log N}{\sqrt{Nh^{p}}}+h^2\right)
$$
with probability at least $1 - 2J_nN^{-C_9} - 2J_nN^{-C_9^{'}}$. By the choice of $C_9$ and $\tau$, we can write $1 - 2J_nN^{-C_9} - 2J_nN^{-C_9^{'}}$ as a sequence $1-S_N$ which converge to 1 in an appropriate rate.
\end{proof}

 \end{lemma}

\section{Additional Graphs}\label{Appendix:additionalgraphs}
\begin{figure}[!h]
    \centering
    \includegraphics[scale=0.8]{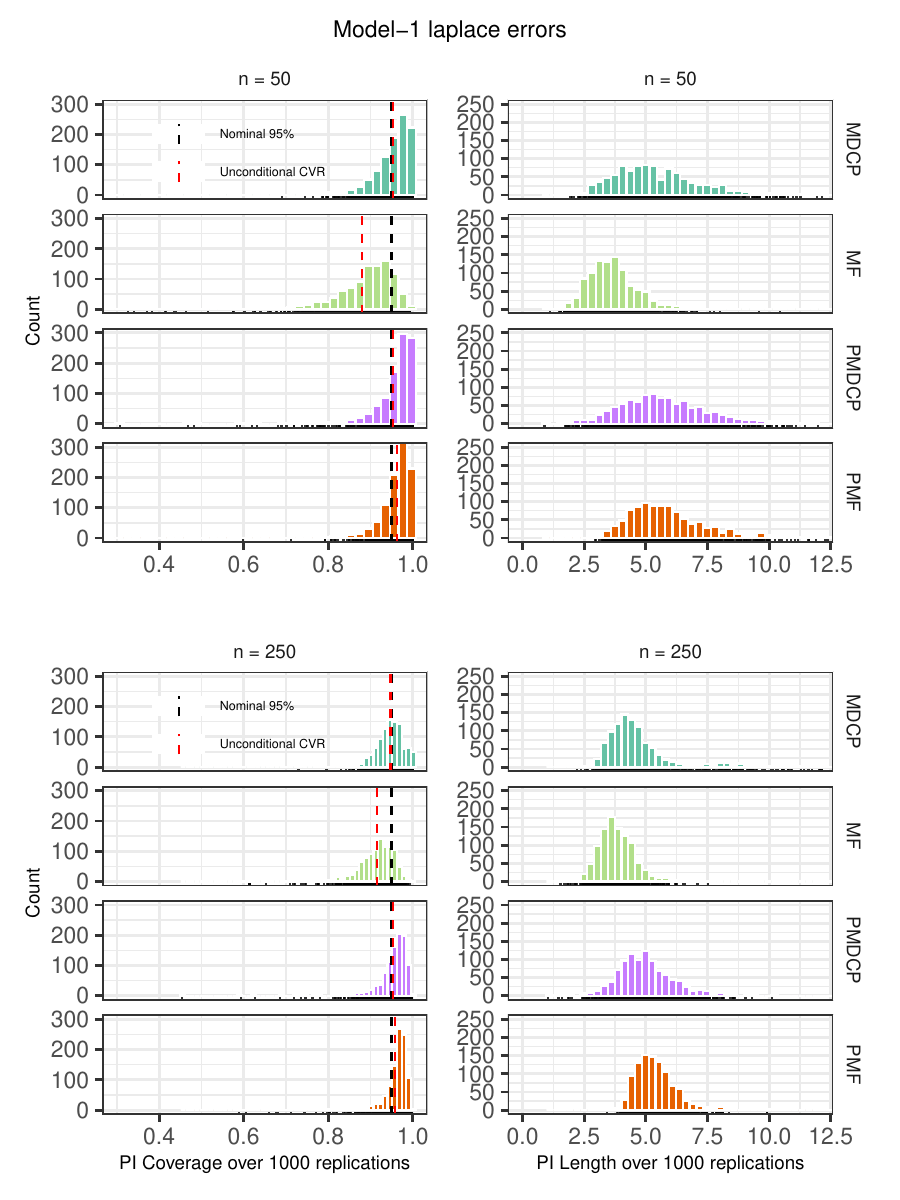}
    \caption{The conditional coverage rate and interval length for different 95$\%$-PIs with simulated data from Model 1 with laplace error.}
    \label{fig:M1laplace}
\end{figure}

\begin{figure}[htbp]
    \centering
    \includegraphics[scale=0.8]{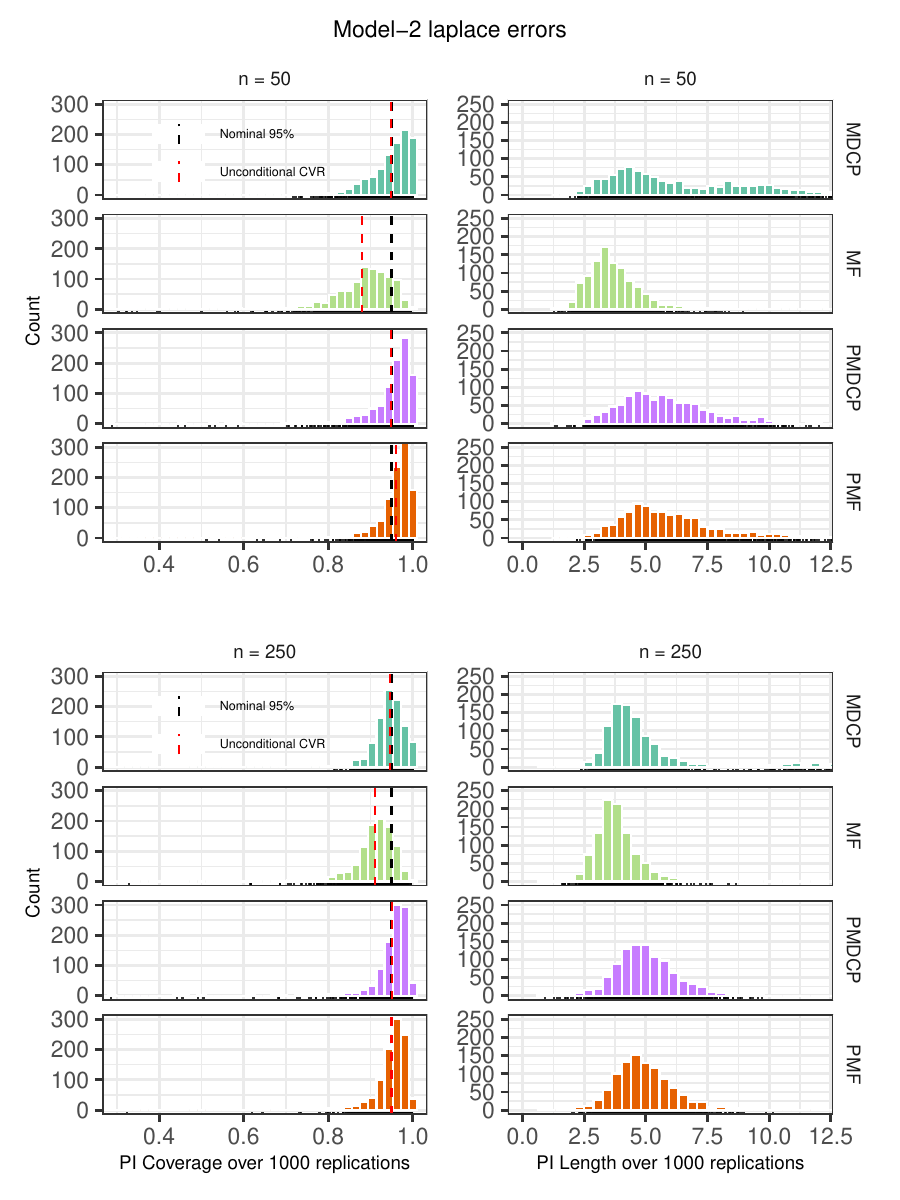}
    \caption{The conditional coverage rate and interval length for different 95$\%$-PIs with simulated data from Model 2 with laplace error.}
    \label{fig:M2laplace}
\end{figure}

%% file: main.bib
@article{shafer2008tutorial,
  title={A tutorial on conformal prediction.},
  author={Shafer, Glenn and Vovk, Vladimir},
  journal={Journal of Machine Learning Research},
  volume={9},
  number={3},
  year={2008}
}

@book{politis2015model,
  title={Model-free prediction and regression: a transformation-based approach to inference},
  author={Politis, Dimitris N},
  year={2015},
  publisher={Springer}
}

@book{li2007nonparametric,
  title={Nonparametric econometrics: theory and practice},
  author={Li, Qi and Racine, Jeffrey Scott},
  year={2007},
  publisher={Princeton University Press}
}

@article{chen2019optimal,
  title={Optimal {Multi-step-ahead} {Prediction} of {ARCH}/{GARCH} {Models} and {NoVaS} {Transformation}},
  author={Chen, Jie and Politis, Dimitris N},
  journal={Econometrics},
  volume={7},
  number={3},
  pages={1--23},
  year={2019},
  publisher={Multidisciplinary Digital Publishing Institute}
}

@book{kosorok2008introduction,
  title={Introduction to empirical processes and semiparametric inference},
  author={Kosorok, Michael R},
  volume={61},
  year={2008},
  publisher={Springer}
}

@article{xu2021conformal,
  title={Conformal prediction for time series},
  author={Xu, Chen and Xie, Yao},
  journal={IEEE transactions on pattern analysis and machine intelligence},
  volume={45},
  number={10},
  pages={11575--11587},
  year={2023},
  publisher={IEEE}
}

@inproceedings{pan2015model,
  title={Model-free bootstrap for markov processes},
  author={Pan, Li and Politis, Dimitris N},
  booktitle={Proceedings of the 60th World Statistics Congress--ISI2015},
  pages={26--31},
  year={2015}
}

@article{sesia2021conformal,
  title={Conformal prediction using conditional histograms},
  author={Sesia, Matteo and Romano, Yaniv},
  journal={Advances in neural information processing systems},
  volume={34},
  pages={6304--6315},
  year={2021}
}

@article{wang2022model,
  title={Model-free bootstrap for a general class of stationary time series},
  author={Wang, Yiren and Politis, Dimitris N},
  journal={Bernoulli},
  volume={28},
  number={2},
  pages={744--770},
  year={2022},
  publisher={Bernoulli Society for Mathematical Statistics and Probability}
}

@inproceedings{zaffran2022adaptive,
  title={Adaptive conformal predictions for time series},
  author={Zaffran, Margaux and F{\'e}ron, Olivier and Goude, Yannig and Josse, Julie and Dieuleveut, Aymeric},
  booktitle={International Conference on Machine Learning},
  pages={25834--25866},
  year={2022},
  organization={PMLR}
}

@article{chudy2020long,
  title={Long-term prediction intervals of economic time series},
  author={Chud{\`y}, Marek and Karmakar, Sayar and Wu, Wei Biao},
  journal={Empirical Economics},
  volume={58},
  number={1},
  pages={191--222},
  year={2020},
  publisher={Springer}
}

@article{wu2024deep,
  title={Deep limit model-free prediction in regression},
  author={Wu, Kejin and Politis, Dimitris N},
  journal={arXiv preprint arXiv:2408.09532},
  year={2024}
}

@article{pan2016bootstrap,
  title={Bootstrap prediction intervals for linear, nonlinear and nonparametric autoregressions},
  author={Pan, Li and Politis, Dimitris N},
  journal={Journal of Statistical Planning and Inference},
  volume={177},
  pages={1--27},
  year={2016},
  publisher={Elsevier}
}

@article{lei2018distribution,
  title={Distribution-free predictive inference for regression},
  author={Lei, Jing and G’Sell, Max and Rinaldo, Alessandro and Tibshirani, Ryan J and Wasserman, Larry},
  journal={Journal of the American Statistical Association},
  volume={113},
  number={523},
  pages={1094--1111},
  year={2018},
  publisher={Taylor \& Francis}
}

@article{lei2013distribution,
  title={Distribution-free prediction sets},
  author={Lei, Jing and Robins, James and Wasserman, Larry},
  journal={Journal of the American Statistical Association},
  volume={108},
  number={501},
  pages={278--287},
  year={2013},
  publisher={Taylor \& Francis}
}

@book{lehmann2005testing,
  title={Testing statistical hypotheses},
  author={Lehmann, Erich Leo and Romano, Joseph P},
  year={2005},
  publisher={Springer}
}

@article{politis2013model,
  title={Model-free model-fitting and predictive distributions},
  author={Politis, Dimitris N},
  journal={Test},
  volume={22},
  number={2},
  pages={183--221},
  year={2013},
  publisher={Springer}
}

@article{chernozhukov2021distributional,
  title={Distributional conformal prediction},
  author={Chernozhukov, Victor and W{\"u}thrich, Kaspar and Zhu, Yinchu},
  journal={Proceedings of the National Academy of Sciences},
  volume={118},
  number={48},
  pages={e2107794118},
  year={2021},
  publisher={National Academy of Sciences}
}

@article{wu2024bootstrap,
  title={Bootstrap prediction inference of nonlinear autoregressive models},
  author={Wu, Kejin and Politis, Dimitris N},
  journal={Journal of Time Series Analysis},
  year={2024},
  volume = {45}, 
  issue = {5},
  page = {800 - 822},
  publisher={Wiley Online Library}
}

@article{politis2023multi,
  title={Multi-Step-Ahead Prediction Intervals for Nonparametric Autoregressions via Bootstrap: Consistency, Debiasing, and Pertinence},
  author={Politis, Dimitris N and Wu, Kejin},
  journal={Stats},
  volume={6},
  number={3},
  pages={839--867},
  year={2023},
  publisher={MDPI}
}

@article{wu2025garchx,
  title={GARCHX-NoVaS: A Bootstrap-Based Approach of Forecasting for GARCHX Models},
  author={Wu, Kejin and Karmakar, Sayar and Gupta, Rangan},
  journal={Journal of Forecasting},
  year={2025},
  publisher={Wiley Online Library}
}

@article{gulay2018comparison,
  title={Comparison of forecasting performances: Does normalization and variance stabilization method beat GARCH (1, 1)-type models? Empirical evidence from the stock markets},
  author={Gulay, Emrah and Emec, Hamdi},
  journal={Journal of Forecasting},
  volume={37},
  number={2},
  pages={133--150},
  year={2018},
  publisher={Wiley Online Library}
}

@article{romano2019conformalized,
  title={Conformalized quantile regression},
  author={Romano, Yaniv and Patterson, Evan and Candes, Emmanuel},
  journal={Advances in neural information processing systems},
  volume={32},
  year={2019}
}

@incollection{efron1992bootstrap,
  title={Bootstrap methods: another look at the jackknife},
  author={Efron, Bradley},
  booktitle={Breakthroughs in statistics: Methodology and distribution},
  pages={569--593},
  year={1992},
  publisher={Springer}
}

@article{chen2016trimmed,
  title={Trimmed conformal prediction for high-dimensional models},
  author={Chen, Wenyu and Wang, Zhaokai and Ha, Wooseok and Barber, Rina Foygel},
  journal={arXiv preprint arXiv:1611.09933},
  year={2016}
}

@article{ibragimov1962some,
  title={Some limit theorems for stationary processes},
  author={Ibragimov, Ildar A},
  journal={Theory of Probability \& Its Applications},
  volume={7},
  number={4},
  pages={349--382},
  year={1962},
  publisher={SIAM}
}

@article{lei2014distribution,
  title={Distribution-free prediction bands for non-parametric regression},
  author={Lei, Jing and Wasserman, Larry},
  journal={Journal of the Royal Statistical Society Series B: Statistical Methodology},
  volume={76},
  number={1},
  pages={71--96},
  year={2014},
  publisher={Oxford University Press}
}

@article{foygel2021limits,
  title={The limits of distribution-free conditional predictive inference},
  author={Foygel Barber, Rina and Candes, Emmanuel J and Ramdas, Aaditya and Tibshirani, Ryan J},
  journal={Information and Inference: A Journal of the IMA},
  volume={10},
  number={2},
  pages={455--482},
  year={2021},
  publisher={Oxford University Press}
}

@inproceedings{vovk2012conditional,
  title={Conditional validity of inductive conformal predictors},
  author={Vovk, Vladimir},
  booktitle={Asian conference on machine learning},
  pages={475--490},
  year={2012},
  organization={PMLR}
}

@book{bosq2012nonparametric,
  title={Nonparametric statistics for stochastic processes: estimation and prediction},
  author={Bosq, Denis},
  volume={110},
  year={2012},
  publisher={Springer Science \& Business Media}
}

@article{zheng2024conformal,
  title={Conformal predictions under Markovian data},
  author={Zheng, Fr{\'e}d{\'e}ric and Proutiere, Alexandre},
  journal={arXiv preprint arXiv:2407.15277},
  year={2024}
}

@book{politis2019time,
  title={Time series: A first course with bootstrap starter},
  author={Politis, Dimitris N and McElroy, Tucker S},
  year={2019},
  publisher={Chapman and Hall/CRC}
}

@article{oliveira2024split,
  title={Split conformal prediction and non-exchangeable data},
  author={Oliveira, Roberto I and Orenstein, Paulo and Ramos, Thiago and Romano, Joao Vitor},
  journal={Journal of Machine Learning Research},
  volume={25},
  number={225},
  pages={1--38},
  year={2024}
}

@article{rosenblatt1952remarks,
  title={Remarks on a multivariate transformation},
  author={Rosenblatt, Murray},
  journal={The annals of mathematical statistics},
  volume={23},
  number={3},
  pages={470--472},
  year={1952},
  publisher={JSTOR}
}

@article{bradley2005basic,
  title={Basic properties of strong mixing conditions. A survey and some open questions},
  author = {Richard C. Bradley},
volume = {2},
journal = {Probability Surveys},
number = {none},
publisher = {Institute of Mathematical Statistics and Bernoulli Society},
pages = {107 -- 144},
year = {2005}
}

@article{bierens1983uniform,
  title={Uniform consistency of kernel estimators of a regression function under generalized conditions},
  author={Bierens, Herman J},
  journal={Journal of the American Statistical Association},
  volume={78},
  number={383},
  pages={699--707},
  year={1983},
  publisher={Taylor \& Francis}
}

@book{billingsley1968convergence,
  title={Convergence of probability measures},
  author={Billingsley, Patrick},
  year={1968},
  publisher={John Wiley \& Sons}
}

@book{potscher1997dynamic,
  title={Dynamic nonlinear econometric models: Asymptotic theory},
  author={P{\"o}tscher, Benedikt M and Prucha, Ingmar},
  year={1997},
  publisher={Springer Science \& Business Media}
}

@article{hormann2010weakly,
  title={Weakly dependent functional data},
  author={H{\"o}rmann, Siegfried and Kokoszka, Piotr},
  year={2010},
  journal = {Annals of Statistics},
  volume = {38},
  number = {3}
}

@article{ibragimov1975independent,
  title={Independent and stationary sequences of random variables},
  author={Ibragimov, I},
  journal={Wolters, Noordhoff Pub.},
  year={1975},
  publisher={Groningen}
}

@article{andrews1984non,
  title={Non-strong mixing autoregressive processes},
  author={Andrews, Donald WK},
  journal={Journal of Applied Probability},
  volume={21},
  number={4},
  pages={930--934},
  year={1984},
  publisher={Cambridge University Press}
}

@article{pascual2006bootstrap,
  title={Bootstrap prediction for returns and volatilities in GARCH models},
  author={Pascual, Lorenzo and Romo, Juan and Ruiz, Esther},
  journal={Computational Statistics \& Data Analysis},
  volume={50},
  number={9},
  pages={2293--2312},
  year={2006},
  publisher={Elsevier}
}

@book{box1976time,
  title={Time Series Analysis: Forecasting and Control},
  author={Box, George EP and Jenkins, Gwilym M and Reinsel, Gregory C and Ljung, Greta M},
  year={2015},
  publisher={John Wiley \& Sons}
}

@article{engle1982autoregressive,
  title={Autoregressive conditional heteroscedasticity with estimates of the variance of United Kingdom inflation},
  author={Engle, Robert F},
  journal={Econometrica: Journal of the econometric society},
  pages={987--1007},
  year={1982},
  publisher={JSTOR}
}

@article{bollerslev1986generalized,
  title={Generalized autoregressive conditional heteroskedasticity},
  author={Bollerslev, Tim},
  journal={Journal of Econometrics},
  volume={31},
  number={3},
  pages={307--327},
  year={1986},
  publisher={North-Holland}
}

@article{pascual2001effects,
  title={Effects of parameter estimation on prediction densities: a bootstrap approach},
  author={Pascual, Lorenzo and Romo, Juan and Ruiz, Esther},
  journal={International Journal of Forecasting},
  volume={17},
  number={1},
  pages={83--103},
  year={2001},
  publisher={Elsevier}
}

@article{chen2004nonparametric,
  title={Nonparametric multistep-ahead prediction in time series analysis},
  author={Chen, Rong and Yang, Lijian and Hafner, Christian},
  journal={Journal of the Royal Statistical Society: Series B (Statistical Methodology)},
  volume={66},
  number={3},
  pages={669--686},
  year={2004},
  publisher={Wiley Online Library}
}

@article{manzan2008bootstrap,
  title={A bootstrap-based non-parametric forecast density},
  author={Manzan, Sebastiano and Zerom, Dawit},
  journal={International Journal of Forecasting},
  volume={24},
  number={3},
  pages={535--550},
  year={2008},
  publisher={Elsevier}
}

@article{guo1999multi,
  title={Multi-step prediction for nonlinear autoregressive models based on empirical distributions},
  author={Guo, Meihui and Bai, Zhidong and An, Hong Zhi},
  journal={Statistica Sinica},
  volume = {9},
  number = {2},
  pages={559--570},
  year={1999},
  publisher={JSTOR}
}

@article{tibshirani2023conformal,
  title={Conformal prediction: Advanced topics in statistical learning (lecture note)},
  author={Tibshirani, R},
  year={2023}
}

@article{angelopoulos2024theoretical,
  title={Theoretical foundations of conformal prediction},
  author={Angelopoulos, Anastasios N and Barber, Rina Foygel and Bates, Stephen},
  journal={arXiv preprint arXiv:2411.11824},
  year={2024}
}

@article{azzalini1981note,
  title={A note on the estimation of a distribution function and quantiles by a kernel method},
  author={Azzalini, Adelchi},
  journal={Biometrika},
  volume={68},
  number={1},
  pages={326--328},
  year={1981},
  publisher={Oxford University Press}
}

@article{jin1999kernel,
  title={On kernel estimation of a multivariate distribution function},
  author={Jin, Zhezhen and Shao, Yongzhao},
  journal={Statistics \& probability letters},
  volume={41},
  number={2},
  pages={163--168},
  year={1999},
  publisher={Elsevier}
}

@article{liu2008kernel,
  title={Kernel estimation of multivariate cumulative distribution function},
  author={Liu, Rong and Yang, Lijian},
  journal={Journal of Nonparametric Statistics},
  volume={20},
  number={8},
  pages={661--677},
  year={2008},
  publisher={Taylor \& Francis}
}

@article{mansouri2024nonparametric,
  title={Nonparametric estimation of bivariate cumulative distribution function},
  author={Mansouri, Behzad and Rastin, Azam and Mombeni, Habib Allah},
  journal={Arabian Journal of Mathematics},
  volume={13},
  number={3},
  pages={621--632},
  year={2024},
  publisher={Springer}
}

@article{mombeni2021asymmetric,
  title={Asymmetric kernels for boundary modification in distribution function estimation},
  author={Mombeni, Habib Allah and Mansouri, Behzad and Akhoond, MohammadReza},
  journal={REVSTAT-Statistical Journal},
  volume={19},
  number={4},
  pages={463--484},
  year={2021}
}

@article{gibbs2025conformal,
  title={Conformal prediction with conditional guarantees},
  author={Gibbs, Isaac and Cherian, John J and Cand{\`e}s, Emmanuel J},
  journal={Journal of the Royal Statistical Society Series B: Statistical Methodology},
  pages={qkaf008},
  year={2025},
  publisher={Oxford University Press UK}
}

@article{pan2016bootstrapmarkov,
  title={Bootstrap prediction intervals for Markov processes},
  author={Pan, Li and Politis, Dimitris N},
  journal={Computational Statistics \& Data Analysis},
  volume={100},
  pages={467--494},
  year={2016},
  publisher={Elsevier}
}

@article{wang2026model,
  title={Model-free bootstrap and conformal prediction in regression: conditionality, conjecture testing, and pertinent prediction intervals},
  author={Wang, Yiren and Politis, Dimitris N},
  journal={Journal of Nonparametric Statistics},
  volume={38},
  number={1},
  pages={311--345},
  year={2026},
  publisher={Taylor \& Francis}
}

@article{chen2012testing,
  title={TESTING FOR THE MARKOV PROPERTY IN TIMESERIES},
  author={Chen, Bin and Hong, Yongmiao},
  journal={Econometric Theory},
  volume={28},
  number={1},
  pages={130--178},
  year={2012},
  publisher={Cambridge University Press}
}

@book{francq2019garch,
  title={GARCH models: structure, statistical inference and financial applications},
  author={Francq, Christian and Zakoian, Jean-Michel},
  year={2019},
  publisher={John Wiley \& Sons}
}
